\documentclass[11pt]{article}
 \expandafter\let\csname equation*\endcsname\relax
 \expandafter\let\csname endequation*\endcsname\relax
\usepackage{bm}
\usepackage{enumerate}
\usepackage[top=1in, bottom=1in, left=1in, right=1in]{geometry}
\usepackage{mathrsfs,color}
\usepackage{amsfonts}\usepackage{multirow}
\usepackage{amssymb}
\usepackage{dsfont}
\usepackage[textsize=tiny]{todonotes}
\usepackage{caption,comment}
\usepackage{blkarray}
\usepackage{amsmath}
\usepackage{float}
\usepackage{amssymb,amsthm}
\usepackage[toc,page]{appendix}
\usepackage{graphicx,afterpage}
\usepackage{epstopdf}
\usepackage{cancel}
\usepackage{mathdots}
\usepackage[pdftex,colorlinks=true]{hyperref}
\usepackage[sort,compress]{cite}

\usepackage{xcolor}
\usepackage{tikz}

\hypersetup{
     colorlinks   = true,
     citecolor    = black,
     linkcolor = black
}

\newtheorem{thm}{Theorem}[section]

\newtheorem{lem}[thm]{Lemma}
\newtheorem{prop}[thm]{Proposition}

\newcommand{\lowO}{\mathbf{\Omega}_L}
\newcommand{\diagO}{\mathbf{\Omega}_D}
\newcommand{\upO}{\mathbf{\Omega}_U}
\newcommand{\Gauss}{\mathbf{G}}
\newcommand{\bfOmega}{\mathbf{\Omega}}
\newcommand{\Cond}{\mathcal{C}}

\newcommand{\bfx}{\mathbf{x}}
\newcommand{\bfK}{\mathbf{K}}
\newcommand{\bfI}{\mathbf{I}}
\newcommand{\bfz}{\mathbf{z}}
\newcommand{\bfA}{\mathbf{A}}
\newcommand{\bfB}{\mathbf{B}}
\newcommand{\bfm}{\mathbf{m}}
\newcommand{\bfy}{\mathbf{y}}
\newcommand{\bfv}{\mathbf{v}}
\newcommand{\bfmhat}{\mathbf{\hat{m}}}
\newcommand{\bfChat}{\mathbf{\widehat{C}}}
\newcommand{\bfH}{\mathbf{H}}
\newcommand{\bfb}{\mathbf{b}}
\newcommand{\bfL}{\mathbf{L}}

\newcommand{\bfh}{\mathbf{h}}
\newcommand{\bfR}{\mathbf{R}}

\newcommand{\E}{\mathbb{E}}

\newcommand{\reals}{\mathbb{R}}

\newcommand{\unit}{\mathds{1}}

\newcommand{\tloc}{\mathrm{loc}}
\newcommand{\tlocp}{\mathrm{p,loc}}

\newcommand{\bfP}{\mathbf{P}}
\newcommand{\bfC}{\mathbf{C}}

\begin{document}
\title{Localization for MCMC: sampling high-dimensional posterior distributions with 
local structure}
\author{M.~Morzfeld, X.T.~Tong, and Y.~Marzouk}
\maketitle

\begin{abstract}

We investigate how ideas from covariance localization in numerical
weather prediction can be used in Markov chain Monte Carlo (MCMC) sampling
of high-dimensional posterior distributions arising in Bayesian inverse problems.  
To localize an inverse problem is to enforce an anticipated ``local'' structure by
(\textit{i}) neglecting small off-diagonal elements of the prior precision and covariance matrices;
and (\textit{ii}) restricting the influence of observations to their neighborhood.
For linear problems we can specify the conditions 
under which posterior moments of the localized problem 
are close to those of the original problem.
We explain physical interpretations of our assumptions
about local structure
and discuss the notion of high dimensionality in local problems,
which is different from the usual notion of high dimensionality in function space MCMC.
The Gibbs sampler is a natural choice of MCMC algorithm for localized inverse problems
and we demonstrate that its convergence rate is independent of dimension
for localized linear problems.
Nonlinear problems can also be tackled efficiently by localization
and, as a simple illustration of these ideas, 
we present a localized Metropolis-within-Gibbs sampler.
Several linear and nonlinear numerical examples 
illustrate localization in the context of MCMC samplers for inverse problems.

\end{abstract}

%

\section{Introduction}
\label{sec:Intro}
We consider inverse problems in the Bayesian setting.
Let $\bfx$ be an $n$-dimensional real-valued random vector.
The observations are defined by 
\begin{equation*}
	\bfy = \bfh(\bfx) +\bfv,
\end{equation*}
where $\bfy$ is a $k$-dimensional vector, 
$\bfh$ is a given function, and $\bfv$
is a random variable with known distribution.
The observations $\bfy$, along with the distribution of $\bfv$, define a likelihood $p_l(\bfy\vert\bfx)$.
A prior probability density $p_0$ describes
prior knowledge about $\bfx$.
For example, one may know that the variables
are likely to be within a certain interval.
The prior distribution often also describes the smoothness of 
a random field whose discretization is the vector~$\bfx$.
The prior and likelihood together define the posterior density:
\begin{equation*}
	p(\bfx\vert \bfy) \propto p_0(\bfx) p_l(\bfy\vert \bfx).
\end{equation*}
Throughout this paper, 
we assume Gaussian errors, $\bfv \sim\mathcal{N}(\bf0,\bfR)$
and Gaussian priors $p_0(\bfx) = \mathcal{N}(\bfm,\bfC)$,
as is common;
see, e.g., \cite{Stuart10}.
In addition, we assume that
\begin{enumerate}[(i)] 
\vspace{-2mm}
\item
the state dimension $n$ is large
and the number of observations $k$ is also large, i.e., $k=O(n)$;
\vspace{-2mm}
\item
the prior covariance and precision matrices are nearly banded;
\vspace{-2mm}
\item
each predicted observation $[\bfh(\bfx)]_j$ 
has significant dependence on only $\ell \ll n$ 
components of $\bfx$ (i.e., the contribution from the remaining components of $\bfx$ is small),
and $\bfR$ is diagonal.   
\end{enumerate}
In simple terms, nearly banded means that the elements away from the diagonal are small,
but we make the meaning of nearly banded and significant dependence more precise below.
We call problems that satisfy the above assumptions ``local.''

In numerical weather prediction (NWP) and ensemble Kalman filtering (EnKF),
local problems arise frequently.
During a typical EnKF step, 
the covariance of $10^8$ variables needs to be estimated accurately,
but the number of samples used is usually 100 or less. 
This seemingly impossible task is made possible by
``localization'' \cite{HoutMitch2001,Hamilletal2001,Houtekamer2005}.
During localization, a nearly banded forecast covariance matrix
is transformed into an (exactly) banded matrix by setting small off-diagonal elements to zero,
e.g., by multiplying each entry of the covariance with a suitable ``localization function'' \cite{GC99}.
In addition, the influence of each observation is restricted to its neighborhood.
Practitioners agree that localization
is a key requirement for  making EnKF applicable to large-scale NWP problems.
Note that localization trades numerical efficiency against errors which can be controlled:
problems with banded forecast covariance structure and local observations are more easily solved, 
and the errors introduced by localization are controllable,
e.g., they vanish when localization thresholds are sufficiently small.
From a more theoretical perspective, it was shown in \cite{BL08} that
one can estimate a 
covariance matrix of bandwidth $l$
with $O(l+\log n)$ samples.
Reference \cite{Tong17} explains how this result is used in the context of EnKF.

This paper examines localization in the context of Bayesian inverse problems
and Markov chain Monte Carlo (MCMC) samplers for the associated posterior distributions.
It is known that sampling \textit{generic} high dimensional (posterior) distributions 
is challenging; see, e.g., \cite{Roberts97, Agapiou16}.
We suggest, however, 
that one can design relatively simple MCMC algorithms to sample
high dimensional posterior distributions of \textit{localized} inverse problems efficiently.
Specifically, we discuss the following three questions:
\begin{enumerate}[(a)]
\vspace{-2mm}
\item
Is the localized problem near the local problem (Section~\ref{sec:Localization})?
\vspace{-2mm}
\item
Can one solve a localized problem efficiently by MCMC (Section~\ref{sec:MCMC})?
\vspace{-2mm}
\item
Are local inverse problems of practical importance (Section~\ref{sec:Discussion})?
\end{enumerate}
\vspace{-1mm}
In Section~\ref{sec:Discussion} we also explain that the notion of high dimensionality of a local inverse problem 
is different from what is usually considered in the MCMC literature \cite{SpantiniEtAl15, CuiEtAl14,Cotter13,BuiEtAl13,FlathEtAl11,PetraEtAl14,CuiEtAl16b}. 
Numerical illustrations are provided in Section~\ref{sec:Examples}.
We summarize our conclusions in Section~\ref{sec:Summary}.

\section{Localization of inverse problems}
\label{sec:Localization}
To localize an  inverse problem means 
to enforce that prior interactions (correlations and\slash or conditional dependencies)
and the effects of observations are confined to a neighborhood.
For local problems,
in the sense of Assumptions (\textit{i})--(\textit{iii}) in Section~\ref{sec:Intro},
one may thus expect that errors introduced by localization are small.
We prove this intuitive result for linear and Gaussian problems,
and then discuss how localization can be used in nonlinear problems.
Note that localization as described here can also be interpreted 
in the context of a more general robust Bayesian analysis \cite{IR12}, which examines how perturbations to the prior and likelihood affect the posterior distribution. 

\subsection{Localization of prior covariance and precision matrices}
\label{sec:LocalizationPriorDetail}
Let $[\bfC]_{i,j}$ be the $i,j$ entry of a covariance matrix $\bfC$.
Suppose that $\vert[\bfC]_{i,j}\vert\ll \vert[\bfC]_{i,i}\vert$ for $\vert i-j\vert\geq l$.
During localization, these ``small'' off-diagonal elements are set to zero.
The resulting localized covariance matrix $\bfC_\text{loc}$ has entries
\begin{equation}
\label{eqn:truncation}
[\bfC_\tloc]_{i,j}=[\bfC]_{i,j} \unit_{|i-j|\leq l},
\end{equation}
where $\unit_{|i-j|\leq l}$ is an indicator function. 
We define the bandwidth $l$ of a $m\times m$ matrix $\bfA$ by
\[
l=\min\{r: [\bfA]_{i,j}=0 \text{ if } |i-j|>r\},
\]
With these definitions and notation,
it becomes clear that localization turns the prior covariance matrix $\bfC$
with small off-diagonal elements 
into a banded matrix $\bfC_\text{loc}$,
whose bandwidth is less or equal to the threshold $l$ used during localization.
Moreover, the localized prior covariance matrix $\bfC_\text{loc}$ is positive definite if 
the minimum eigenvalue of $\bfC$ is above $\delta_C$,
where 
\begin{equation}
\label{eq:deltaC}
	\delta_C:=\max_i \sum_{j: |i-j|>l} |[\bfC]_{i,j}|.
\end{equation}
We show in Appendix~A (Proposition \ref{prop:loc}) that
\begin{equation}
\label{eq:CPertCov}
\|\bfC-\bfC_\tloc\|\leq \delta_C,
\end{equation}
i.e., the localized prior covariance matrix is a small perturbation
of the prior covariance matrix.
Note that, throughout this paper, we use $\|\,\cdot\,\|$ to denote the $l_2$ norm for a vector
$\bfx$ with elements $x_1,\dots,x_n$,
$\|\,\bfx\,\|=\sqrt{\sum_{j=1}^n x_j^2}$,
as well as the $l_2$ operator norm for a matrix $\bfA$, 
$\|\,\bfA\,\|=\sup_{v\in \reals^n, \|v\|=1}\|\bfA\bfv\|$.

Working with prior precision matrices,
rather than prior covariance matrices,
is sometimes more natural.
In this case, one can write the prior distribution as 
$p_0(\bfx) = \mathcal{N}(\bfm,\bfOmega^{-1})$,
where $\bfOmega$ is the prior precision matrix.
Precision matrices can be localized
in the same way as covariance matrices.
We assume, as above, that 
$\vert[\bfOmega]_{i,j}\vert\ll \vert[\bfOmega]_{i,i}\vert$ for $\vert i-j\vert\geq l$
and localize the prior precision matrix by setting
small off-diagonal elements equal to zero.
The result is a localized prior precision matrix with entries
\[
	[\bfOmega_\tloc]_{i,j}=[\bfOmega]_{i,j} \unit_{|i-j|\leq l}.
\]
Under our assumptions, 
the localized precision matrix is a small perturbation of the precision matrix:
\begin{equation}
\label{eq:deltaOmega}
\|\bfOmega-\bfOmega_\tloc\|\leq \delta_\Omega, \quad
\delta_\Omega:=\max_i \sum_{j: |i-j|>l} |[\bfOmega]_{i,j}|.
\end{equation}
Localization of the covariance or precision matrices results in localized prior distributions
 $p_{0,\tloc}(\bfx) = \mathcal{N}(\bfm,\bfC_\tloc)$ or $p_{0,\tloc} =\mathcal{N}(\bfm,\bfOmega^{-1}_\tloc)$.

\subsection{Localization of the observation matrix}
\label{sec:LocObs}
We first consider the case $\bfh(\bfx) = \bfH\bfx$,
where $\bfH$ is a given $k\times n$ matrix.
Assumption~(ii) in Section~\ref{sec:Intro}
implies that each observation $[\bfH\bfx]_j$ may depend
on all components of $\bfx$,
but the contributions of many components are negligible.
In this case, 
the observation matrix $\bfH$ can be
localized similarly to how we localized the prior covariance.

For a given observation matrix $\bfH$ and a given threshold~$l_H$, define 
a localized observation matrix~by
\[
[\bfH_\tloc]_{j,i}=[\bfH]_{j,i} \unit_{|o_j-i|\leq l_H}.
\]
where  $o_j$  represents the ``center'' 
of the $j$-th observation. 
Following \eqref{eq:deltaC}, we can quantify the difference between $\bfH$ and $\bfH_\tloc$ by 
\begin{equation}
\label{eq:deltaH}
	\delta_H=\max_i\left\{\sum_{j: |i-o_j|>l_H}|[\bfH]_{j,i}|, \sum_{i: |i-o_j|>l_H}|[\bfH]_{j,i}|\right\}. 
\end{equation}	
As before, if $\delta_H$ is small, then errors due to the localization are expected to be small.

\subsection{Localized posterior distributions of linear--Gaussian inverse problems}
We continue to assume that $\bfh(\bfx) = \bfH\bfx$,
so that the true (original)  posterior distribution is a Gaussian with mean and covariance given by
\begin{align*}
	\bfmhat &= \bfm+\bfK \left( \bfy-\bfH \bfx \right),\\
	\bfChat &= \left(\bfI -\bfK\bfH\right) \bfC,\\
	\bfK &= \bfC\bfH^T(\bfR+\bfH\bfC\bfH^T)^{-1},
\end{align*}
where $\bfK$ is the Kalman gain.

Localization of the observation matrix leads to the localized likelihood
$p_{\text{loc},l}(\bfy\vert\bfx)=\mathcal{N}(\bfH_\tloc \bfx,\bfR)$.
If we also localize the prior covariance matrix,
we obtain the localized prior $p_{0,\tloc}(\bfx) = \mathcal{N}(\bfm,\bfC_\tloc)$.
The localized prior and likelihood then define the \textit{localized posterior}
$p_\text{loc}(\bfx\vert\bfy)\propto p_{\text{loc},l}(\bfy\vert\bfx)p_{\text{loc},0}(\bfx)$,
whose mean and covariance are given by 
\begin{align*}
	\bfmhat_\tloc &= \bfm+\bfK_\tloc \left( \bfy-\bfH_\tloc \bfx \right),\\
	\bfChat_\tloc &= \left(\bfI -\bfK_\tloc\bfH_\tloc\right)\bfC_\tloc, \\
	\bfK_\tloc &= \bfC_\tloc\bfH_\tloc^T(\bfR+\bfH_\tloc\bfC_\tloc\bfH_\tloc^T)^{-1}.\nonumber
\end{align*}
We prove 
in Appendix~A (Proposition \ref{prop:loc})
that the means and covariance matrices of the localized and original posterior distributions satisfy 
\begin{align*}
\|\bfmhat-\hat{\bfm}_\tloc\|\leq & 
(\delta_C+\delta_H)\cdot D_1(\|\bfR^{-1}\|, \|\bfChat\|, \|\bfC^{-1}\|)\cdot (\|\bfm\|+\|\bfy\|), \\
\|\bfChat-\bfChat_\tloc\|\leq& (\delta_C+\delta_H) \cdot D_2(\|\bfR^{-1}\|, \|\bfChat\|, \|\bfC^{-1}\|),
\end{align*}
where the functions $D_1$ and $D_2$ are defined in Proposition \ref{prop:loc}. 
Similarly, if we work with the prior precision matrix,
we obtain, after localization, the prior $p_{0,\tloc}(\bfx) = \mathcal{N}(\bfm,\bfOmega^{-1}_\tloc)$,
which leads to a localized posterior whose mean and covariance are given by 
\begin{align}
\label{eq:LocPostMean}
	\bfmhat_{\tlocp} &= \bfm+\bfK_{\tlocp} \left( \bfy-\bfH_\tloc \bfx \right),\\
\label{eq:LocPostCov}	
	\bfChat_{\tlocp} &= (\bfI-\bfK_{\tlocp}\bfH_{\tloc})\bfOmega_\tloc^{-1},\\
	\bfK_{\tlocp} &= \bfOmega_\tloc^{-1}\bfH_\tloc^T
	(\bfR+\bfH_\tloc\bfOmega_\tloc^{-1}\bfH_\tloc^T)^{-1}.
\end{align} 
We find in Proposition \ref{prop:locprecision} that:
\begin{align*}
\|\bfmhat-\hat{\bfm}_{\tlocp}\|\leq & 
(\delta_\Omega+\delta_H)\cdot D_3(\|\bfR^{-1}\|, \|\bfChat\|, \|\bfOmega\|)\cdot (\|\bfm\|+\|\bfy\|), \\
\|\bfChat-\bfChat_{\tlocp}\|\leq& (\delta_\Omega+\delta_H) \cdot D_4(\|\bfR^{-1}\|, \|\bfChat\|, \|\bfOmega\|),
\end{align*}
where $D_3$ and $D_4$ are defined in Proposition \ref{prop:locprecision}. 

In summary, the difference between the localized and
unlocalized posterior distributions depends on the localization thresholds we chose,
on the structure of the prior covariance or precision matrix,
as well as on the observation matrix $\bfH$ and its localization.
The bandwidth of localizations of $\bfC$ or $\bfOmega$ and $\bfH$ should therefore be tuned
to obtain small differences between the localized and unlocalized problems.

We use a threshold for localization
because it makes our proofs simpler.
In practice, one may localize more effectively using suitable localization functions \cite{GC99},
which set small off-diagonal elements to zero smoothly
and which can preserve positive-definiteness during localization.
Moreover, we only consider (nearly) banded covariance matrices,
which arise, for instance, in linear--Gaussian Bayesian inverse problems
on one-dimensional spatial domains.
The conceptual ideas of localization, however, can be used 
in problems with 2D spatial domains (see Section~\ref{sec:Examples}
for a numerical example).
In fact, we anticipate that many of our ideas can be adapted to matrices 
with more general sparsity patterns 
but defer this investigation to future work. 

\subsection{Localization of nonlinear problems}
\label{sec:LocNonlin}
In a nonlinear inverse problem,
the function $\bfh$ is not linear, and thus the posterior distribution is non-Gaussian 
even when the prior distribution is Gaussian.
However, the function $\bfh$ is usually well understood
because it is the result of a careful modeling effort.
Thus, it is not unreasonable to assume that it is known
whether $\bfh$ is local
in the sense of Assumption (\textit{iii}) in Section~\ref{sec:Intro}.
If $\bfh$ is local, 
then we can localize it by neglecting the some of the components of $\bfx$ 
when computing the components, $[\bfh]_i$, of $\bfh$
(see Section~\ref{sec:lMwGObsLoc} for more detail on how to do this).
The localized $\bfh$, along with the localized Gaussian prior,
defines the localized posterior distribution of a nonlinear inverse problem. 
The localized posterior distribution may be close to the posterior distribution of the unlocalized problem,
but we do not prove this statement in the general, nonlinear setting.
The lack of theoretical results, however, does not prevent us
from \textit{using} localization in nonlinear problems,
and we present such an example in Section~\ref{sec:Examples}.
Moreover, more than a decade of experience with using localization in EnKF and NWP 
can also be viewed as numerical and empirical evidence
that localization is indeed applicable in nonlinear problems
(see also \cite{MHS17}).

\section{MCMC for local inverse problems}
\label{sec:MCMC}
MCMC is often used for the numerical solution of Bayesian inverse problems.
To illustrate the behavior of some MCMC algorithms on localized problems,
we first consider the extreme case of a linear problem with \textit{diagonal}
covariance and precision matrix and with a diagonal 
observation function $\bfh(\bfx) = \bfx$.
Specifically, suppose the target distribution is 
the $n$-dimensional Gaussian distribution $p(x) = \mathcal{N}(0,\bfI)$, 
where $\bfI$ is the identity matrix of  dimension~$n$.
Suppose that the current state of the Markov chain is $\bfx^k$.
The Metropolis-Hastings (MH) algorithm proposes a move to $\bfx'$
by drawing from a proposal distribution $q(\bfx'\vert \bfx^{k})$,
and accepts or rejects the move with probability
\begin{equation*}
	a^{k+1} =\min\left\{1, \frac{p(\bfx')q(\bfx^{k}\vert \bfx')}{p(\bfx^{k})q(\bfx'\vert \bfx^{k})}\right\};
\end{equation*} 
see, e.g., \cite{Owen,MacKay97,Kalos86}.
Averages over the samples generated in this way converge to expected values 
with respect to the target distribution~$p$ as $k\to\infty$.
A question of practical importance is:
how many samples are needed to 
accurately estimate expectations with respect to the target distribution?
The answer depends on the proposal distribution.
For a Gaussian proposal distribution
$q(\bfx^{k+1}\vert \bfx^k) = \mathcal{N}(\bfx^k,\sigma^2\bfI)$, which
produces a so-called random walk Metropolis (RWM) chain,
one must choose a proposal variance $\sigma^2$ such that the acceptance probability
is reasonably large while, at the same time,
the accepted MCMC moves are large enough to explore the space appropriately.
An optimal choice that achieves
this trade-off for RWM is $\sigma^2=O(n^{-1})$, 
where~$n$ is the dimension of the problem;
see, e.g., \cite{Beskos09,Roberts97,Roberts98}.
Optimal scalings of the proposal variance with problem
dimension are also known for other MCMC algorithms.
For example, the proposal distribution of the Metropolis-adjusted Langevin algorithm (MALA)
is defined by
\begin{equation*}
\bfx' = \bfx^k+\frac{\sigma^2}{2}\nabla \log p(\bfx^k)+\sigma\, \xi
\end{equation*}
where $\nabla$ denotes a gradient and $\xi$ is a vector
of $n$ standard normal variates.
An optimal choice is $\sigma=O(n^{-1/3})$.
For Hamiltonian Monte Carlo (see, e.g., \cite{Neal11,Duane87})
an optimal step size is $\sigma=O(n^{-1/4})$~\cite{Beskos13}
(here and below we refer to $\sigma$ as a step size
and to $\sigma^2$ as the proposal variance).

Setting aside issues of transient behavior \cite{christensen2005scaling},
the efficiency of an MCMC algorithm can be assessed by
computing the integrated auto-correlation time (IACT).
Throughout this paper we use the definitions and 
numerical approximations of IACT discussed in \cite{Wolff04}.
Heuristically,  the number of effective samples
is the number of samples divided by IACT;
see, e.g., \cite{Kalos86,MacKay97}.
In Figure~\ref{fig:Iso} we compute IACT for various MCMC algorithms
applied to the $n$-dimensional isotropic Gaussian,
as a function of $n$.
\begin{figure}[t]
\centering
\includegraphics[width=.6\textwidth]{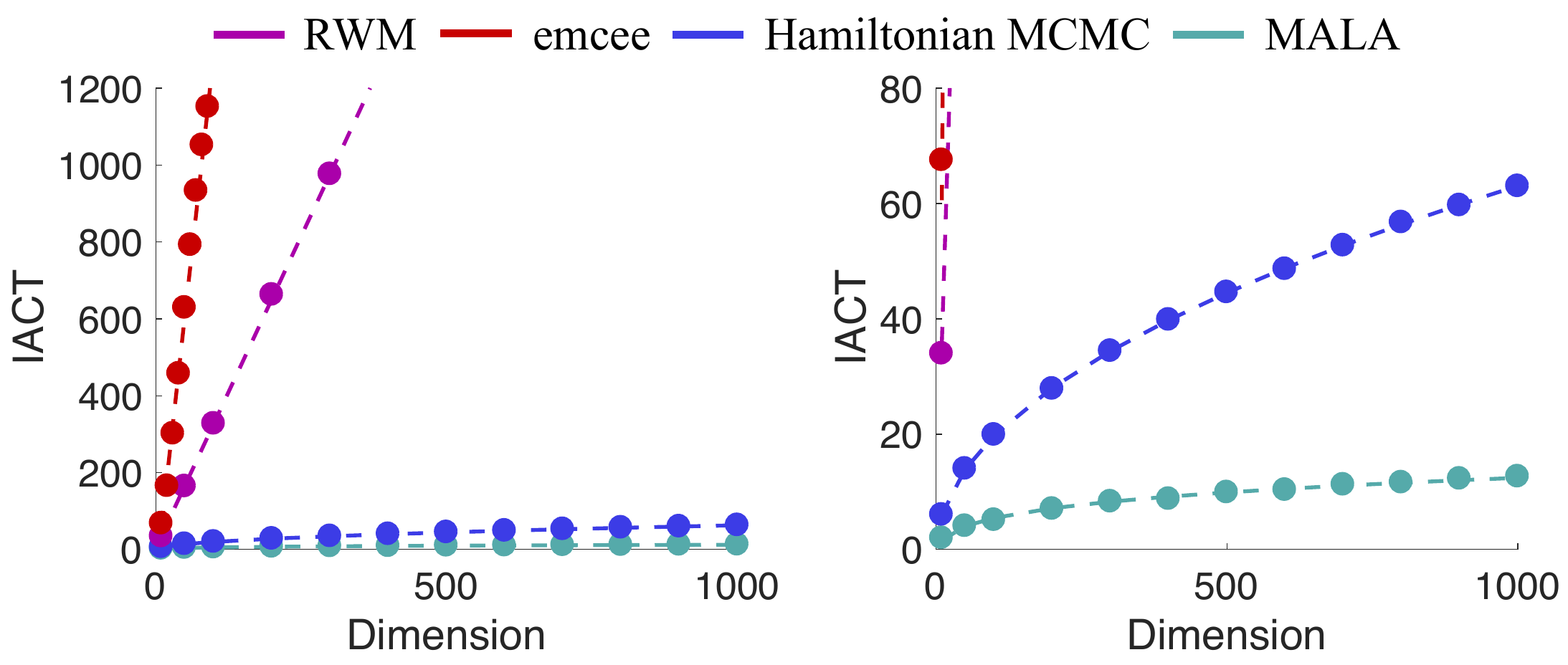}
\caption{
IACT as a function of dimension for various MCMC samplers
applied to an isotropic Gaussian.
Dots represent IACT, averaged over all $n$ variables,
for emcee (red), RWM (purple),
Hamiltonian MCMC (blue),
and MALA (teal).
The dashed lines represent linear fits (emcee and RWM),
fits to a square root (Hamiltonian MCMC),
and a fit to a $1/3$-degree polynomial (MALA).
}
\label{fig:Iso}
\end{figure}
We observe that IACT grows with dimension for 
all algorithms we consider, but at different rates.
For both RWM and an affine invariant sampler \cite{GoodmanWeare10,Hammer} 
called emcee (or the MCMC Hammer), 
IACT grows linearly with the dimension $n$;
for Hamiltonian MCMC, we observe that IACT grows with the square root of $n$,
while for MALA, IACT grows as~$n^{1/4}$.
Similar tests were done for the ``$t$-walk'' \cite{tWalk},
a general purpose ensemble sampler.
The numerical results of \cite{tWalk}
suggest a linear scaling of $t$-walk's IACT with dimension.

However, sampling an isotropic Gaussian is trivial,
and MCMC should be independent of dimension
for this problem because it can be decomposed into
$n$ independent sub-problems
(see also \cite{Rebeschini15,MHS17}).
An MCMC sampler that naturally makes use 
of local problem structure is the Gibbs sampler
(sometimes Gibbs samplers are also called ``heat bath'' or ``partial resampling,''
and these can be viewed as examples of ``single-component Metropolis algorithms,'' 
see, e.g., \cite{MCMCInPractice}).
The basic Gibbs sampler is as follows.
Let the target density be $p({\bfx})$, where $\bfx$ is shorthand
for the vector with $n$ elements $x_1,x_2,\dots,x_n$.
Set $j=1$, and let $\bfx^k$ be the current state of a Markov chain.
Then a Gibbs sampler proceeds as follows:
\begin{enumerate}[(i)]
\vspace{-2mm}
\item 
Set $k\to k+1$.
\vspace{-2mm}
\item
Sample $x_{j}^{k+1}$ from the conditional distribution 
$p(x_j\vert x_1^{k+1},\dots x_{j-1}^{k+1},x_{j+1}^{k},\dots, x_{n}^{k})$.
\vspace{-2mm}
\item
Repeat (ii) for all $n$ elements $x_j$ of $\bfx$.
\end{enumerate}
\vspace{-1mm}
Repeating this process $N_e$ times, one obtains samples such that
averages over these samples converge to expected values 
with respect to the target distribution $p$ as $N_e\to\infty$.

The Gibbs sampler generates independent samples, independently of dimension,
for the isotropic Gaussian
(an extreme example of a local problem).
Similarly, a block Gibbs sampler 
generates independent samples, independently of dimension,
if the covariance or precision matrices are block-diagonal.
This suggests that MCMC based on Gibbs samplers may be more effective
\textit{for local problems}
than the MCMC algorithms we considered above.
In this section, we investigate this idea in more detail and study
convergence rates of Gibbs samplers for Gaussian distributions
with banded covariance and precision matrices.

Many of the results we present below may be known,
but we decided to summarize what is important about Gibbs sampling for 
our purposes because 
(\textit{i}) the Gibbs sampler and its effectiveness in local
problems is essential to the understanding of how MCMC can function in high-dimensional problems
with local structure;
(\textit{ii}) we could not find references on the connection 
between MCMC convergence rates and local problem structure,
perhaps because relevant results are spread over several papers and books
in different disciplines (applied mathematics, physics, and statistics) and over several decades.

\subsection{Dimension independent convergence rates of Gibbs samplers}
\label{sec:Gibbs}
For simplicity, we assume that there are $m$ blocks of the same size $q$ so that $n=mq$.
We divide the $n$ elements of $\bfx$ according to the $m$ blocks
and write $\bfx=(\bfx_1,\cdots, \bfx_m)$,
with the understanding that each $\bfx_j$ consists of $q$ consecutive elements in $\bfx$. 
A blocked Gibbs sampler uses the conditionals defined for each $\bfx_j$,
i.e., at the $k$th step, 
we sample the block $\bfx_j$ using the $q$-dimensional conditional 
$p(\bfx_j\vert \bfx_1^{k+1},\dots \bfx_{j-1}^{k+1},\bfx_{j+1}^{k},\dots, \bfx_{m}^{k})$.
We call the resulting algorithm a ``Gibbs sampler with block-size $q$''
because the ``standard'' Gibbs sampler above is a Gibbs sampler of block-size one in this terminology. 

We first consider a Gibbs sampler with block-size $q$
for Gaussian distributions with $q$-block-tridiagonal covariance matrices.
We say that a matrix $\mathbf{A}$ is $q$-block-tridiagonal, if 
\[
\mathbf{A}_{i,j}=\mathbf{0}\quad \text{for}\quad (i,j)\notin \{(i,i),(i,i+1), (i,i-1), i=1,\cdots, m\},
\]
where $\bfA_{i,j}$ is the $(i,j)$-th $q\times q$ block of~$\bfA$.
It is straightforward to see that a matrix with bandwidth $l$ is $l$-block-tridiagonal, and a $q$-block-tridiagonal matrix has bandwidth less than $2q$. 
Figure \ref{fig:blocking} illustrates a $4$-block-tridiagonal matrix with bandwidth $l=4$.
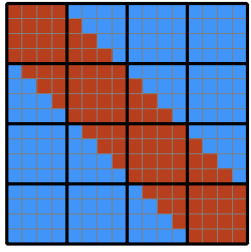
\begin{figure}[tb]
\begin{center}
\begin{tikzpicture}[scale=0.4]
\definecolor{locblue}{RGB} {65,149,249}
\definecolor{locred}{RGB} {184,63,29}
\fill[locblue] (-2,-2) rectangle (6,6);

\fill[locred] (-2,6) rectangle (0,4);
\fill[locred] (-1,5) rectangle (1,3);
\fill[locred] (-1.5,5.5) rectangle (.5,3.5);
\fill[locred] (0,4) rectangle (2,2);
\fill[locred] (-0.5,4.5) rectangle (1.5,2.5);
\fill[locred] (1,3) rectangle (3,1);
\fill[locred] (0.5,3.5) rectangle (2.5,1.5);
\fill[locred] (2,2) rectangle (4,0);
\fill[locred] (1.5,2.5) rectangle (3.5,0.5);
\fill[locred] (3,1) rectangle (5,-1);
\fill[locred] (2.5,1.5) rectangle (4.5,-0.5);
\fill[locred] (4,0) rectangle (6,-2);
\fill[locred] (3.5,0.5) rectangle (5.5,-1.5);

\draw[step=0.5cm,gray,very thin] (-2,-2) grid (6,6);
\draw[step=2cm,black,very thick] (-2.01,-2.01) grid (6,6);

\end{tikzpicture}
\end{center}
\caption{Sparsity pattern of a $4$-block-tridiagonal matrix with bandwidth $l=4$.
The red color indicates nonzero entries and the blue color indicates zero entries.
The black squares define the $q=4$ blocks $\bfC_{i,j}$.}
\label{fig:blocking}
\end{figure}
The convergence rate of the Gibbs sampler is dimension-independent
if the covariance matrix is $q$-block-tridiagonal,
as detailed in the following theorem.
\begin{thm}
\label{thm:gsnsimple}
Suppose the Gibbs sampler with block-size $q$ is applied to a Gaussian target distribution 
$p=\mathcal{N}(\bfm, \bfC)$ with $m$ blocks of size $q$.
Suppose $\bfC$ is $q$-block-tridiagonal.
Then the distribution of $\bfx^k$ converges to $p$ geometrically fast in all coordinates,
and we can couple $\bfx^k$ and a sample $\bfz\sim \mathcal{N}(\bfm,\bfC)$ such that 
\begin{equation}
\label{eqn:convergence}
\E\|\bfC^{-1/2}(\bfx^k-\bfz)\|^2\leq \beta^k n(1+\|\bfC^{-1/2} (\bfx^0-\bfm)\|^2),
\end{equation}
where
\begin{equation}
\label{eqn:betacov}
\beta\leq \frac{2(1-\Cond^{-1})^2\Cond^4}{1+2(1-\Cond^{-1})^2\Cond^4},
\end{equation}
with $\Cond$ being the condition number of $\bfC$. 
\end{thm}

Similarly, the convergence rate of the Gibbs sampler
is dimension-independent 
if the precision matrix, rather than the covariance matrix, is 
block-tridiagonal.
One can modify Theorem~\ref{thm:gsnsimple}
to address the case of banded precision matrices. 
\begin{thm}
\label{thm:gsnsimplePrecision}
Suppose the Gibbs sampler with block-size $q$
is applied to a Gaussian target distribution $p=\mathcal{N}(\bfm, \bfOmega^{-1})$ 
with $m$ blocks of size $q$. Suppose $\bfOmega$ is $q$-block-tridiagonal.
Then the distribution of $\bfx^k$ converges to $p$ geometrically fast in all coordinates,
and we can couple $\bfx^k$ and a sample $\bfz\sim \mathcal{N}(\bfm,\bfOmega^{-1})$ such that 
equation~(\ref{eqn:convergence}) holds with
\begin{equation}
\label{eqn:betaprecision}
\beta\leq \frac{\Cond(1-\Cond^{-1})^2}{1+\Cond(1-\Cond^{-1})^2},
\end{equation}
where $\Cond$ is the condition number of $\bfOmega$.
\end{thm}
\noindent
The proofs of Theorems~\ref{thm:gsnsimple} and~\ref{thm:gsnsimplePrecision} 
can be found in Appendix~A.

While upper bounds for the rate of convergence $\beta$
are independent of the dimension,
the upper bounds themselves are linear in $n$ and 
involve a norm of the initial condition, which may also scale linearly in $n$.
This does not cause practical difficulties
because the dimension independent scaling of the convergence \textit{rate} implies
that the number of iterations required to reach a given error level
scales logarithmically in $n$, 
which is essentially a constant in practice
(even when $n$ is large).
For example, suppose one wants that the $l^2$ error
$\E\|\bfC^{-1/2}(\bfx^k-\bfz)\|^2$ be bounded by a threshold $\varepsilon$.
To reach this goal, the Gibbs sampler must perform $k$ iterations, where
\[
\beta^k n(1+\|\bfC^{-1/2} (\bfx^0-\bfm)\|^2)\leq \epsilon\quad \Rightarrow\quad k\geq  \frac{\log n-\log \epsilon+\log(1+ \|\bfC^{-1/2} (\bfx^0-\bfm)\|^2) }{-\log \beta}.
\]

Finally, we emphasize that dimension-independent convergence for linear problems is not a new observation, and that there are multiple ways to accelerate this convergence \cite{Fox, GS89}. 
Yet it is often assumed that the spectral gap of an associated linear operator, 
e.g., the Gauss-Seidel operator, is dimension independent.  
Our theoretical contribution here is to show that this dimension-independent gap indeed exists when the inverse problem is local.

\subsubsection{Banded covariance matrices vs.~banded precision matrices}
\label{sec:PrecVSCov}
The upper bounds on the convergence rates
depend on the condition number of the covariance matrix,
or, equivalently, on the condition number of the precision matrix.
This means that convergence is fast, in any dimension,
only if this condition number is moderate.
However, if the condition number 
of the covariance matrix is moderate, 
a banded covariance matrix implies that the precision matrix
can be approximated by a banded matrix, and vice versa. 
This equivalence is made precise in Lemma 2.1 of \cite{BL12}. 
Thus, our assumptions and Theorems~\ref{thm:gsnsimple} and~\ref{thm:gsnsimplePrecision} 
describe and apply to one unified class of problems:
convergence of the Gibbs sampler is fast,
in any dimension, if the covariance \textit{and} precision matrices are banded,
i.e., if statistical interactions (correlations and conditional dependencies) are local. 

\subsubsection{Computational costs}
\label{sec:computcost}
While the upper bound of the convergence rate is independent of dimension,
the actual computational cost of sampling may not be independent of dimension.
The Gibbs sampler for Gaussians with banded covariance matrix
requires  matrix square roots and linear solves
of matrices of sizes $m\times m$ and $d\times d$, 
where $d = n-m$.
These square roots need only be computed once (not once per sample),
but the cost of this computation increases with~$n$ at a rate that 
depends on the bandwidth of the covariance matrix.
Generating a sample requires, at each block, 
solution of a banded linear problem,
and some matrix-vector multiplications
and vector-vector operations.
The cost-per-sample is dominated by the linear solve,
and this cost also increases with~$n$ and at a rate  that 
depends on the bandwidth of the covariance matrix.
The computational cost for one sample is thus,
roughly, $m$ times the cost of the linear solve.

If the precision matrix is banded, 
the conditional distributions used during Gibbs sampling simplify
and it is sufficient to condition on neighboring blocks
(see, e.g., equation~\eqref{tmp:xstep} in the Appendix).
The required matrix operations (square roots, linear solves)
depend on the size of the blocks $q$,
but not on the number of blocks.
Assuming that $q\ll m\ll n$,
the computational cost of a Gibbs sampler for a Gaussian
with banded precision matrix is, roughly, $m$
times the cost of computations with blocks of size~$q$.

Our discussion of the Gibbs sampler so far has applied to generic (localized) Gaussian targets. 
Thus the Gibbs sampler can, in principle, be used to draw samples from posterior distributions of linear inverse problems.
Assuming that the prior covariance or precision matrices are localized,
the localized \textit{posterior} covariance and precision matrices are also banded,
because, as shown in \cite{BL12}, 
the basic arithmetic operations in \eqref{eq:LocPostCov} 
preserve bandedness. 
The convergence rate of the Gibbs sampler
is independent of dimension in this case.
However, linear inverse problems can be solved by a variety of other specialized MCMC techniques; 
see, e.g., \cite{BuiEtAl13,Fox}.
We do not claim that the Gibbs sampler
is necessarily a competitive computational strategy for large scale (million or more variables in $\bfx$)
linear inverse problems with local structure,
but we do anticipate that effective samplers
 can be built from a combination of localization, Gibbs sampling,
acceleration, and preconditioning methods.

\subsection{Metropolis-within-Gibbs for localized nonlinear inverse problems}
\label{sec:lMwG}
In Section~\ref{sec:Gibbs} we described  the Gibbs sampler for generic Gaussian target distributions with banded covariance and precision matrices.
When $\bfh(\bfx)$ is nonlinear, however, the posterior distribution is not Gaussian and in general may not have tractable full conditionals, which makes the direct use of Gibbs sampling infeasible. 
In this section we continue to assume that the prior is Gaussian and use the Gibbs sampler to draw from this prior and then Metropolize. This simple Metropolis-within-Gibbs (MwG) sampler can handle nonlinear $\bfh$ and samples non-Gaussian posterior distributions. The sampler can be localized (l-MwG) by localizing the Gaussian prior and the likelihood.

\subsubsection{Localized Metropolis-within-Gibbs sampling: banded covariance}
\label{sec:inv}
Suppose the prior covariance matrix is block-tridiagonal (after localization),
with $m$ blocks of size $q$,
and further suppose that $\bfx^k$
is the current state of the Markov chain.
One iteration of the l-MwG sampler is as follows.
Start with the first of $m$ blocks.
Use the Gibbs proposal with block-size $q$ (see Section~\ref{sec:Gibbs})
to propose a local move $\bfx'$ by drawing a sample from the 
localized Gaussian prior,
conditioned on the current state $\bfx^k$.
Accept or reject the (local) move by taking the observations into account,
i.e., accept with probability
\begin{equation}
\label{eq:AccProb}
	a = \min\left\{1, \frac{\exp\left(-0.5\left(\bfy-\bfh(\bfx')\right)^T\bfR^{-1}\left(\bfy-\bfh(\bfx')\right)\right)}
	{\exp\left(-0.5\left(\bfy-\bfh(\bfx)\right)^T\bfR^{-1}\left(\bfy-\bfh(\bfx)\right)\right)}\right\},
\end{equation}
Iterating these steps over all $m$ blocks completes one move of l-MwG.
The l-MwG sampler converges to the localized posterior distribution.
This follows from the usual theory of Metropolis-within-Gibbs sampling.

\subsubsection{Localized Metropolis-within-Gibbs sampling: banded precision}
The l-MwG sampler 
can also be applied if the precision matrix,
rather than the covariance matrix, is given.
The sampler is as described above,
but the implementation using precision matrices can be numerically more efficient. 
If the precision matrix has bandwidth $l$, 
then the conditional distribution of a block depends only on a few neighboring blocks,
so that the computations required for drawing a sample from the 
conditional distributions require only matrices of size much less than $n$
(see also Section~\ref{sec:PrecVSCov}, and equation~\eqref{tmp:xstep} in the appendix).
This conditional independence also implies that one can 
sample the block independently of other far away blocks. 
This provides opportunities for leveraging parallel computing 
to reduce overall wall-clock time.

\subsubsection{Localization of the likelihood}
\label{sec:lMwGObsLoc}
The l-MwG sampler requires that $\bfh(\bfx)$ be evaluated
for each block even though only a small number of the components of $\bfx$ are 
changed during one of the local moves.
Since we assume that $\bfh(\bfx)$ is a local function,
i.e., each element of $\bfh(\bfx)$ depends only on a few components of $\bfx$,
one may want to use the local structure of $\bfh(\bfx)$ during sampling.
This localization can accelerate the computations of $\bfh(\bfx)$,
and can also increase acceptance rates and shorten the burn-in period.
Yet changing $\bfh(\bfx)$ alters the likelihood and, therefore, 
the posterior distribution of the inverse problem.
For linear $\bfh$, we showed in Section~\ref{sec:Localization} that this change in the posterior distribution can be small.
We now show how to localize the likelihood under the assumption that the localized posterior distribution can be written as
\begin{equation}
	\label{eqn:qinverse}
	p_\tloc(\bfx\vert\bfy)\propto p_{0,\tloc}(\bfx)\exp\left(-\sum_j  H_j(\bfx_{I_j}, \bfy_j)\right),
\end{equation}
where the observations $\bfy_j$ are the observations ``assigned'' to 
the set of variables  $\bfx_{I_j}$,
 $p_{0,\tloc}$
is the localized Gaussian prior distribution, and the functions $H_j$ can be nonlinear. 

Recall that the Gibbs sampler for the Gaussian prior
produces a local update $\bfx'$ from $\bfx$ 
such that $\bfx_j=\bfx'_j$ for $j\neq i$.
Localization of the likelihood means to accept the proposed local adjustment of the sample, $\bfx'$,
with probability 
\begin{equation}
\label{eq:ARCriterion}
\min\left\{1, \frac{\exp(-\sum_{j: i\in I_j}  H_j(\bfx'_{I_j}, \bfy_j))}{\exp(-\sum_{j: i\in I_j}  H_j(\bfx_{I_j}, \bfy_j))}\right\},
\end{equation}
i.e., we only use the observations assigned to the block we are sampling 
when we consider acceptance of the move.

The stationary distribution of the l-MwG sampler
is the localized posterior distribution \eqref{eqn:qinverse}.
To prove this statement,
it suffices to check the detailed balanced relation 
\[
p_\tloc(\bfx\vert\bfy)Q_i(\bfx,\bfx')=p_\tloc(\bfx'\vert\bfy)Q_i(\bfx',\bfx)
\]
for $\bfx\neq\bfx'$,  where $Q_i$ is the transition density resulting
from the above two steps (propose a local sample using the Gibbs proposal for the prior,
then accept or reject it using the local criterion~(\ref{eq:ARCriterion})). 
For $\bfx\neq\bfx'$, the transition density has the explicit form:
\[
Q_i(\bfx,\bfx')=K_i(\bfx,\bfx')\min\left\{1, \frac{\exp(-\sum_{j: i\in I_j}  H_j(\bfx'_{I_j}, \bfy_j))}{\exp(-\sum_{j: i\in I_j}  H_j(\bfx_{I_j}, \bfy_j))}\right\}.
\]
where the transition density $K_i(\bfx,\bfx')$ is defined by the Gibbs move.
The Gaussian prior $p_{0,\tloc}(\bfx)$ is 
the invariant distribution of this transition, in the sense that 
\[
p_{0,\tloc}(\bfx)K_i(\bfx,\bfx')=p_{0,\tloc}(\bfx')K_i(\bfx',\bfx).
\]
Without loss of generality, we assume $\sum_{j: i\in I_j}  H_j(\bfx'_{I_j}, \bfy_j)\leq\sum_{j: i\in I_j}  H_j(\bfx_{I_j}, \bfy_j)$
which leads to
\begin{align*}
p_\tloc(\bfx\vert\bfy)Q_i(\bfx,\bfx')&=p_{0,\tloc}(\bfx)K_i(\bfx,\bfx')\exp\left(-\sum_{j:i\in I_j}  H_j(\bfx_{I_j}, \bfy_j)\right),\\
p_\tloc(\bfx'\vert\bfy)Q_i(\bfx',\bfx)&=p_{0,\tloc}(\bfx')K_i(\bfx',\bfx) \exp\left(-\sum_{j:i\in I_j}  H_j(\bfx_{I_j}, \bfy_j)\right)\frac{\exp(-\sum_{j: i\in I_j}  H_j(\bfx_{I_j}, \bfy_j))}{\exp(-\sum_{j: i\in I_j}  H_j(\bfx'_{I_j}, \bfy_j))}. 
\end{align*}
Because we also have that
\[
p_{0,\tloc}(\bfx)K_i(\bfx,\bfx')=p_{0,\tloc}(\bfx')K_i(\bfx',\bfx),\quad \bfx'_j=\bfx_j\quad \forall j\neq i,
\]
detailed balance,
i.e., $p_\tloc(\bfx\vert\bfy)Q_i(\bfx,\bfx')=p_\tloc(\bfx'\vert\bfy)Q_i(\bfx',\bfx)$,
is now verified.

\subsubsection{Computational requirements of l-MwG}
Recall that the convergence rate of the Gibbs sampler for
the localized prior is independent of dimension
(Theorems~\ref{thm:gsnsimple} and~\ref{thm:gsnsimplePrecision}).
This means that the size of the proposed moves 
does not decrease as the dimension of $\bfx$ increases,
since the size of the move depends only on the local properties of 
the prior in one of the blocks, rather than the overall number of blocks.
In contrast, the step sizes of many other MCMC algorithms
decrease with the dimension of $\bfx$ (see Figure~\ref{fig:Iso}).
Moreover, the proposed l-MwG moves are local,
i.e., $\bfx'$ is different from $\bfx$ only in 
a few (much less than $n$) of its components,
independent of the number of observations $k$
and of the dimension $n$ of~$\bfx$.
Because we assume that $k=O(n)$,
the number of observations grows 
with the dimension, but the 
number of the observations per block
remains fixed. This suggests that the acceptance ratio in equation~(\ref{eq:AccProb})
may be independent of dimension.

If the size of the proposed moves and the acceptance ratio do not decrease with dimension, 
then the convergence rate of the Gibbs sampler could be independent of dimension.
Indicators of sampling efficiency, e.g., IACT,
would depend only on the \textit{properties} of the blocks 
rather than on the \textit{number} of blocks,
with properties of each block  defined by 
the local structure of the prior covariance and precision matrices
and local properties of the observation function.
The overall computational cost per sample of l-MwG 
is then linear in the number of blocks.
We do not provide a proof for the
dimension independence convergence of l-MwG,
but we provide numerical examples (linear and nonlinear)
in which IACT is indeed independent of dimension,
while IACT increases with dimension for other MCMC algorithms.
We also present an example in which we deliberately violate
some of our assumptions to demonstrate the limitations of these ideas.

\subsubsection{Limitations of l-MwG}
We note that l-MwG has the format of ``sample from the prior and correct (Metropolize) to account for the likelihood.'' 
This approach is generally not efficient and, in particular, degenerates as the observational noise diminishes and the posterior concentrates with respect to the prior.
Localization can somewhat mitigate the effect of posterior concentration (at least, relative to non-localized samplers) as Metropolization is applied only to low-dimensional blocks.
Nonetheless, practical MCMC samplers will require more sophisticated proposal distributions \textit{and} localization.
The l-MwG presented here should be viewed as a first and simple example of an MCMC sampler
that can achieve dimension independent performance by exploiting underlying local structure.
We remark that l-MwG using `` likelihood-informed proposals''
is feasible by mixing localization with other MCMC strategies,
such as more sophisticated proposal distributions \cite{GC11},
acceleration by matrix splittings and analogies of Gibbs samplers to linear solvers (see \cite{Fox}),
multigrid methods \cite{GS89}, or preconditioning. 
This is beyond the scope of this paper, and will be investigated in the future.

\section{Discussion of assumptions and effective dimension}
\label{sec:Discussion}
\subsection{Physical interpretation of local problems}
\label{sec:physical}
We have shown, for linear problems, that localization causes
small errors if a problem is local (see Assumptions (\textit{i}--\textit{iii})
in Section~\ref{sec:Intro}).
We argued in Section~\ref{sec:MCMC} that localized problems can be solved efficiently by MCMC.
All this is relevant only if there are indeed inverse problems
which are local and which can be localized,
i.e., if Assumptions (\textit{i}--\textit{iii})
in Section~\ref{sec:Intro} are valid for some interesting problems.

The assumptions in Section~\ref{sec:Intro} correspond to prior distributions with short
correlation lengths (e.g., in a Gaussian process) and small
neighborhood sizes (e.g., in a Gaussian Markov random field
\cite{rue2005gaussian}) relative to the dimensions of the physical domain.  
A common choice are Gaussian prior distributions with precision
matrices defined via Laplace-like operators; see, e.g.,
\cite{Stuart10,lindgren2011explicit,MartinEtAl12,BuiEtAl13,PetraEtAl14,CuiEtAl16}
and our examples in Section~\ref{sec:Examples}.  
These priors typically have banded precision matrices and, 
in many cases, also have (nearly) banded covariance matrices. 
The priors are updated to posterior distributions by
likelihoods involving observations that depend largely on local properties.  
This means in particular that the set of observations, $\bfy$,
may inform \textit{all} components of $\bfx$,
but each individual component of the observation, $[\bfy]_j$, $j=1,\dots,k$,
may only inform a subset of the components of $\bfx$.
An example of a local observation function is when
some or all components of the ``quantity of interest'' $\bfx$ are directly observed, 
which is often the case in state estimation problems. 

We expect that the assumptions are valid in many
geophysical and engineering applications, where the target
distribution is the posterior distribution of a physical quantity
defined over a spatial domain. 
For example, observation matrices $\bfH$ in image deblurring are often
constructed through the discretization of kernels that have (nearly)
compact support, and, for that reason, are typically local.  
Observations of diffusion processes are local on sufficiently short time scales.
PDEs where information is carried mostly along characteristics
(e.g., transport equations) give rise to local observations when
the quantity of interest is an initial condition.
We already brought up NWP and the EnKF 
as an example of a local problem in which localization 
enables efficient computations in high dimensional problems.
Similarly, exploiting localization in importance sampling (particle filtering)
is also a current topic in NWP, see, e.g.,
\cite{LeiBickel11,Penny16,PJvL15,Rebeschini15,Poterjoy15,Poterjoy16,Poterjoy17,
Todter15,Reich13,Lee16}.
NWP, however, is usually not considered an inverse problem
due to its sequential-in-time nature.

It is important to realize that many important problems are not local, and that localization is not useful for such problems.
An example is computed tomography,
where each observation might depend on material properties along an entire
tomographic ray, leading to observation matrices which are not local in the sense we define here.

\subsection{Connections with infinite dimensional inverse problems and effective dimensions}
\label{sec:InfDimProbs}
One can think of $\bfx$ as the discretization of a physical quantity in some domain, 
for instance on a grid with $n$ degrees of freedom or in Fourier basis of $n$ modes.
For a given domain, the dimension~$n$ of $\bfx$ grows as the discretization is refined. 
If the number of observations $k$ is held constant and we let $n\to\infty$, 
then we describe what happens as the discretization is refined while the domain and observation network remain fixed. 
This leads to the concept of an effective dimension, 
which may be small (finite) even when the apparent dimension is large (infinite); see, e.g, \cite{CM13,Agapiou16}.

Related to a small effective dimension are low-rank updates 
from prior to posterior distributions \cite{SpantiniEtAl15, CuiEtAl14, FlathEtAl11}.
A low-rank update means that the posterior distribution differs from the 
prior distribution in $n_\text{eff}\ll n$ directions. 
More precisely, for Gaussian problems,
a low-rank update means that 
the difference between prior and posterior covariance is low rank.
In practice this occurs, for example, 
when the number of observations is much less than the dimension of $\bfx$,
or when the observations constrain only a few linear combinations of the components of $\bfx$.
Indeed, some definitions of effective dimension
are directly related to the difference of prior and posterior covariance.
In \cite{Agapiou16}, an effective dimension is defined by: 
\begin{equation}
\label{eqn:effdim1}
	n_{\text{eff},1} = \text{tr}((\bfC - \hat{\bfC})\bfC^{-1}),
\end{equation}
where $\bfC$ and $\hat{\bfC}$ are prior and posterior covariance, respectively, and where $\text{tr}(\cdot)$ is the trace of a matrix.
Alternatively, one can consider precision matrices and define an effective dimension by
\begin{equation}
\label{eqn:effdim2}
n_{\text{eff},2} = \text{tr}((\hat{\bfOmega} - \bfOmega)\bfOmega^{-1}),
\end{equation}
where $\bfOmega$ and $\hat{\bfOmega}$
are the prior and posterior precision matrices, respectively.
Other effective dimensions have been defined 
for specific importance sampling algorithms;
see, e.g., \cite{CM13,Agapiou16}.
Certain MCMC and importance sampling algorithms can exploit a small effective
dimension or low rank updates, and can be made discretization
invariant---such that indicators of computational efficiency become independent of the chosen grid; 
see, e.g., \cite{Cotter13,BuiEtAl13,FlathEtAl11,PetraEtAl14,CuiEtAl16b}.

Our assumptions and problem setting describe a different mechanism for reaching large dimensions, 
because we assume that the number of observations is on
the order of the dimension, $k=O(n)$.
The assumptions translate to a problem for which 
the discretization and the number of observation per unit ``length'' remain fixed, 
while the size of the domain grows.
As an example, consider a Gaussian process
defined on an interval of length $L$. 
Further suppose that this process is observed every $r$ units of distance,
i.e., we have $L/r$ observations.
We investigate the situation where the domain size $L$
gets bigger but the number of observations per unit length remains constant.
The limit $L\to\infty$ may not be meaningful
because it would correspond to an infinitely large domain.
Considering large $L$, however, is meaningful
because it describes a large domain
and a large number of observations.
Moreover, as $L$ increases,
the dimension $n$, the number of observations $k$,
and the effective dimension all increase.
The cartoon example of sampling a high-dimensional isotropic Gaussian
illustrates the performance one can expect
from MCMC when the size of the domain is large,
and the number of observations is on the order of the dimension
(i.e., large, well-observed domains).

The assumptions defining a local problem (see Section~\ref{sec:Intro})
do not in general imply that updates from prior to posterior are low-rank. 
Rather, we replace assumptions about low effective dimension
and\slash or low-rank updates by 
assumptions about locality of the precision\slash covariance matrices
and local observations.
The Gibbs samplers discussed above are efficient
when precision and covariance matrices are banded,
and if observations are local.
If these assumptions are only approximately satisfied,
we propose to localize the problem,
i.e., to replace covariance\slash precision matrices by banded matrices.
The resulting localized problem can be solved efficiently,
but at the cost of additional errors due to the localization procedure.

In summary, we suggest that while MCMC 
for generic high-dimensional problems might remain difficult,
one can effectively solve two classes of problems.
If an effective dimension is small and\slash or an update from 
prior to posterior distribution is low rank,
then one can exploit this structure and create effective
MCMC algorithms.
This case has been discussed extensively over the past years.
Our contribution is to suggest that MCMC can also be made efficient if
effective dimensions are large or if updates from prior to posterior distributions
are high rank, as long as the prior covariance and precision matrices are banded
and the observations are local.
In fact, the situation we describe is analogous to linear algebra.
High-dimensional matrices are easy to deal with if they are either
low-rank or banded (sparse).
The same seems true in MCMC for Bayesian inverse problems---these
problems are manageable if either an effective dimension is small, or
if the covariance and precision matrices have banded structure
and if the observations are local.

\section{Numerical illustrations}
\label{sec:Examples}
We consider several linear and nonlinear examples to illustrate local inverse problems,
their localization, as well as the dimension independence of Gibbs
and l-MwG samplers.

\subsection{Example 1: Gaussian prior with exponential covariance function}
We consider a Gaussian prior on the interval $z\in\left[0,L\right]$
with mean and covariance function
\begin{equation*}
\mu(z) = 5\,\sin(2\pi z),\quad
k(z,z') = C\,\exp\left(-\frac{\vert z-z'\vert }{2\rho}\right).
\end{equation*}
where $\rho=0.02$, $C=10$,
and fix a discretization of the domain with $\Delta z=0.01$.
This results in a slowly varying prior mean and,
for a given discretization, 
the covariance matrix has off-diagonal elements that
decay quickly away from the diagonal.

For a given domain size $L$, the dimension of the discretized Gaussian 
is $n=L/\Delta z$, i.e., we have 100 state variables per unit length.
For numerical stability we add $10^{-6}\,\mathbf{I}$ to the discretized covariance matrix.
Samples from the prior are shown (in teal) for a problem with $L=2$ 
in Figure~\ref{fig:AnisoIllu}.
\begin{figure}[t]
\centering
\includegraphics[width=.4\textwidth]{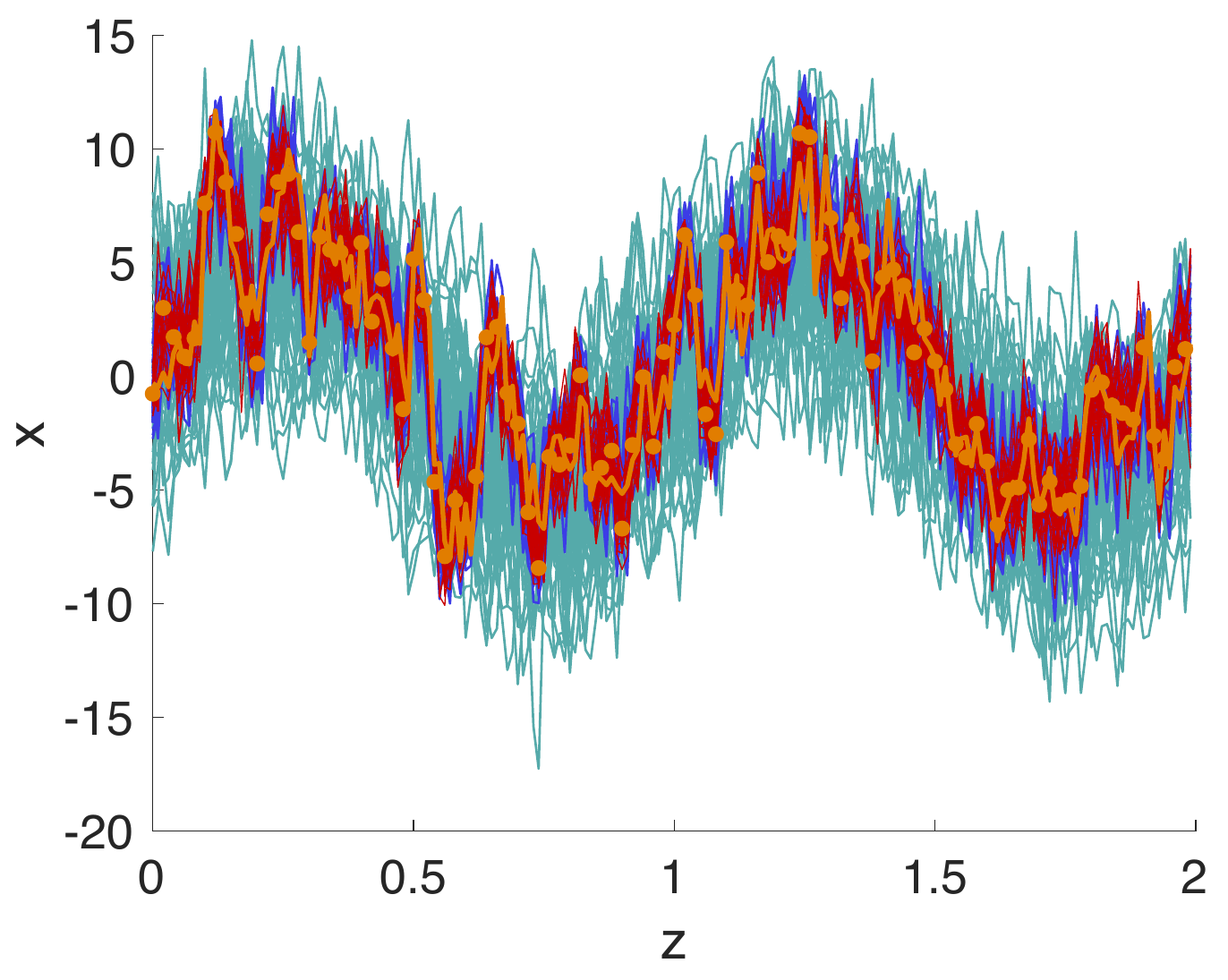}
\caption{
Illustration of the Gaussian inverse problem of example~1.
Teal -- 50 samples of the prior distribution.
Red -- 50 samples of l-MwG.
Blue (often hidden) -- 50 samples of posterior distribution.
Orange line -- ``true state'' that gives rise to the data.
Orange dots -- data.
}
\label{fig:AnisoIllu}
\end{figure}
The data are measurements of a prior sample,
collected every $2\Delta z$ length units.
The measurements are perturbed by Gaussian noise with mean zero
and identity covariance matrix, i.e.,
$\bfR={\bf{I}}$, the identity matrix of size $n/2$.
Note that the observation network is such that no localization of 
the observation matrix $\bfH$ is required,
because the observations are already local (point-wise measurements
and diagonal $\bfR$).
For a given $L$ we have $n$ variables to estimate and $n/2$ data points,
or, equivalently, we have 100 variables and 50 measurements per unit length.
The true state and measurements are illustrated (in orange) for a problem 
with $L=2$ in Figure~\ref{fig:AnisoIllu}.
The data and prior distribution define a Bayesian posterior distribution,
and we show 50 samples of this posterior distribution (for $L=2$) in Figure~\ref{fig:AnisoIllu}.

To illustrate the banded structure and relative size of the elements of the
prior and posterior covariance\slash precision matrices,
we plot the absolute value of the elements of these matrices,
scaled by their largest element, 
in the left panels of Figure~\ref{fig:Prior1}.
These scaled absolute values are between 0 and 1,
and indicate where off-diagonal elements are small.
\begin{figure}[t]
\centering
\includegraphics[width=.4\textwidth]{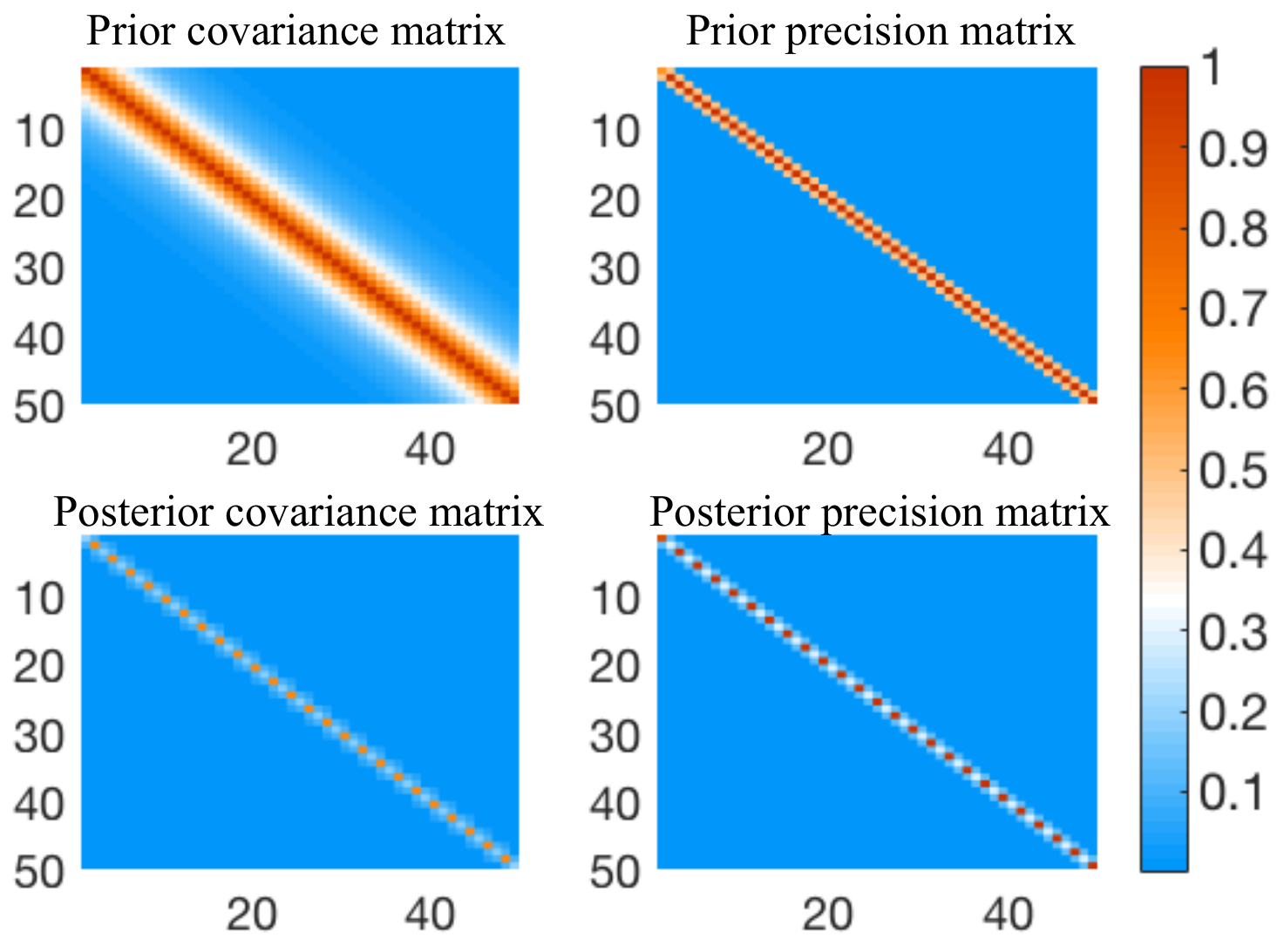}
\includegraphics[width=.4\textwidth]{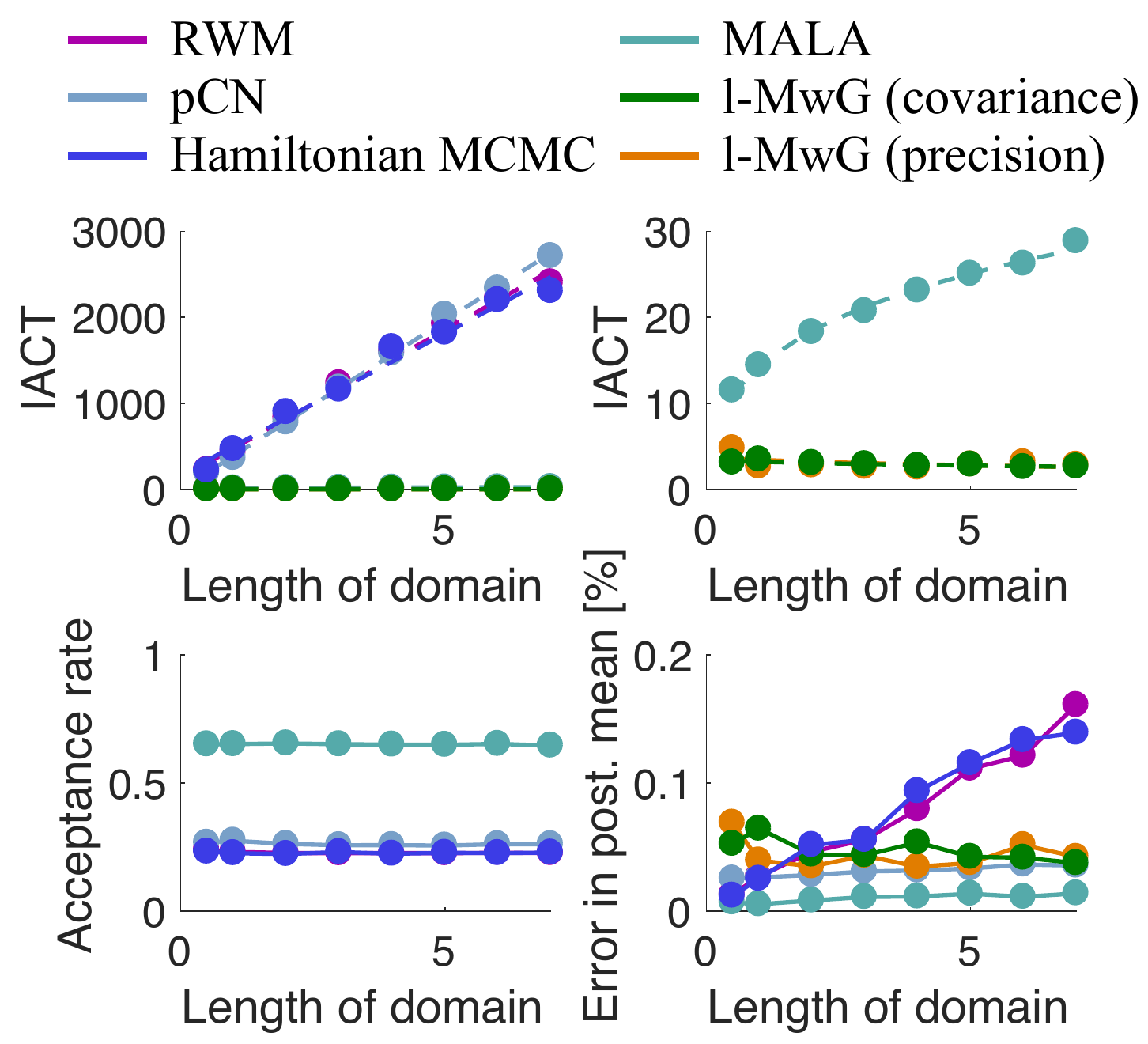}
\caption{
Covariance and precision matrices, IACT, acceptance rate
and errors for Example~1.
\textit{Left panels}. 
Prior and posterior covariance and precision matrices of example~1 with $L=0.5$.
Top row: prior covariance (left) and precision matrix (right).
Bottom row: posterior covariance (left) and precision matrix (right).
\textit{Right panels}.
Top row: IACT as a function of domain length for
RWM, pCN, Hamiltonian MCMC (left)
and MALA and l-MwG (right).
Bottom row: acceptance rate as a function of domain length (left)
and mean squared error in posterior mean (right).
}
\label{fig:Prior1}
\end{figure}
The prior and posterior covariance matrix
have nearly banded structure in the sense that the elements near the diagonal
are larger than the off-diagonal elements.
The condition number of the prior precision and covariance matrices is about 64
(independently of the domain length).
We localize the prior covariance matrix by setting 
all elements smaller than $0.1$ equal to zero.
Since the prior precision matrix is tridiagonal,
we do not need to localize it.
The l-MwG sampler in ``precision matrix implementation''
already makes use of the banded structure of the precision matrix
because we condition only on a few neighbors, rather than the full state.
Since the observation matrix $\bfH$ is already local (direct observations of the state variables),
$\bfH$ does not need to be localized. 
The likelihood has the form~(\ref{eqn:qinverse})
and we make use of this structure
by considering blocks of size $2$ with one observation in each block.

For a fixed domain length $L$,
we apply RWM, pCN \cite{Cotter13}, MALA, Hamiltonian MCMC, and l-MwG
(with covariance localization or in precision matrix implementation).
We tune the proposal variance of RWM, pCN, MALA, Hamiltonian MCMC
by considering a scaling of the step size with 
$n^{-k}$, for several different values of $k$, e.g.,
$k=\{1,2,3,4,5,6\}$ (recall that the dimension of our discretization is $n=L/\Delta z$).
We call the step size that lead to the smallest IACT ``optimal.''
All chains are initialized by a random sample from the prior
and we produce $10^6$ samples with pCN and RWM, 
$10^5$ with Hamiltonian MCMC,
$10^3$ with MALA, and 500 with the l-MwG samplers.
The set-up for this problem is also summarized in Table~\ref{tab:Prior1}.
\begin{table}[tb]
\begin{center}
\caption{Configurations and IACT scalings for examples 1,2 and 3.
In each cell, left to right, are the configurations and results for examples 1,2 and 3.}
\label{tab:Prior1}
\begin{tabular}{l c c c}
Algorithm				& Tuned step size			& Sample size					& IACT scaling\\
\hline
RWM				&$n^{-1/2},n^{-1/2},n^{-1/2}$	&	$10^6,10^6,10^6$			& $n,n,n$ \\
MALA				&$n^{-1/6},n^{-1/4},n^{-1/4}$	&	$10^4,10^4,10^4$			& $n^{1/4},n^{1/4},n^{1/4}$ \\
Hamiltonian MCMC		&$n^{-1/2},n^{-1/2},n^{-1/3}$	&	$10^5,10^5,10^5$			& $n,n,n$ \\
pCN					&$n^{-1/2},n^{-1/2},n^{-1/3}$	&	$10^6,10^6,10^6$			& $n,n,n$ \\
l-MwG (precision)		&n/a						&	$500,500,5000$			& const.,const.,$n^{1/2}$\\
l-MwG (covariance)		&n/a						&	$500,500,5000$			& const.,const.,$n^{1/2}$ \\
\end{tabular}
\end{center}

\end{table}

We perform numerical experiments for domain lengths
ranging from $L=0.5$ to $L=7$,
which leads to problems of dimension $n=50$ to $n=700$,
and with $25$ to $350$ observations.
For each domain length and algorithm we compute the corresponding IACT.
This leads to the scalings, obtained by least-squares fitting of polynomials, 
of IACT with dimension,
as shown in Table~\ref{tab:Prior1},
and as illustrated in Figure~\ref{fig:Prior1}.
We observe that RWM, pCN, Hamiltonian MCMC,
and MALA exhibit an increasing IACT (at different rates) as the domain
size and the number of observations increase,
while both implementations of l-MwG are characterized by an IACT
that remains constant.

We further compute the acceptance rate
(i.e., the acceptance ratio, averaged over all moves)
for RWM, pCN, Hamiltonian MCMC, and MALA.
Results are shown in Figure~\ref{fig:Prior1},
and we note that our tuning of the step-size for each algorithm
keeps the acceptance rate constant as dimension increases.
Finally, we compute an error to check that each algorithm
indeed samples the correct distribution.
We define an ``error in posterior mean'' by
\begin{equation}
\label{eq:Error}
	e = \frac{1}{n}\sum_{j=1}^n ([\hat\bfm]_j-[\bar\bfx]_j)^2,
\end{equation} 
where $\hat\bfm$ is the posterior mean 
and where $\bar\bfx$ is the average over the MCMC samples.
We compute this error, which describes how far our estimated
posterior mean is from the actual posterior mean,
for each algorithm and show the results in Figure~\ref{fig:Prior1}.
We note that all algorithms lead to a small error.
Since this is also true for l-MwG with prior covariance localization,
we conclude that the localized problem is indeed nearby the unlocalized problem we set out to solve.

In this example, it is not the overall dimension,
the overall number of observations, or the rank of the update of
prior to posterior covariance matrix that defines performance bounds for l-MwG.
Because the local problem structure is used during problem formulation and during its MCMC solution, 
the characteristics of each loosely coupled block define the behavior of the sampler.
The overall number of blocks is irrelevant.

One may wonder why the IACT of pCN increases with dimension,
even though pCN is ``by design'' dimension independent.
The reason is \textit{how} dimension increases in this example
(see the discussion in Section~\ref{sec:InfDimProbs}).
The pCN algorithm is dimension independent if 
the dimension increases
because a discretization is refined,
while the size of the domain, 
the number of observations, and
an effective dimension remain constant.
In the current example,
pCN is dimension independent for fixed domain size $L$
and a fixed observation network (fixed $k$),
while we decrease the discretization parameter $\Delta z$.
In such a scenario, an effective dimension remains constant.
However, as we increase the domain size $L$ \textit{and} the number of observations $k$,
while keeping the discretization parameter $\Delta z$ fixed,
the apparent and effective dimensions both increase, and the performance of pCN deteriorates 
(see left panel of Figure~\ref{fig:EffDimEx1} and also Section~\ref{sec:InfDimProbs}).
Similarly, the number of non-negligible
eigenvalues of prior or posterior covariance matrices increases with 
the domain size and the number of observations (see right panel of Figure~\ref{fig:EffDimEx1}).
The usual scenario considered for pCN (and other function space MCMC algorithms)
is that the number of ``relevant'' eigendirections remains constant as dimension increases.
This, too, is not the case for this example and may not be true for other
local inverse problems.
\begin{figure}[t]
\centering
\includegraphics[width=.9\textwidth]{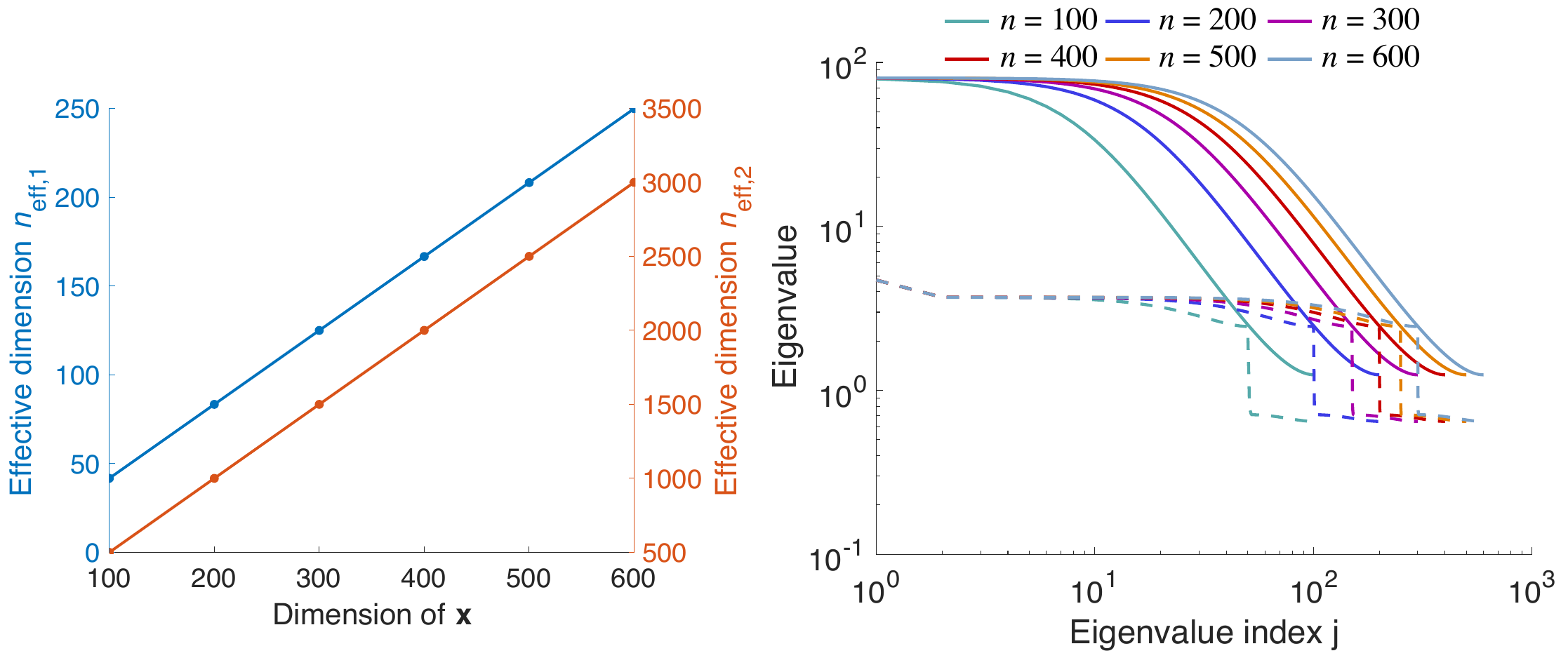}
\caption{
Covariance and precision matrices, IACT, acceptance rate
and errors for Example~1.
\textit{Left panel}. 
Effective dimensions (see Section~\ref{sec:InfDimProbs})
as a function of dimension (equivalently domain size
and number of observations).
\textit{Right panel}.
Eigenvalues of prior (solid)
and posterior (dashed) covariance matrices
for domains of sizes between $L=1$ and $L=6$.
}
\label{fig:EffDimEx1}
\end{figure}

The original convergence analysis of pCN can be found in~\cite{hairer2014spectral}, and applies to a setting where dimension increases because the discretization of the unknown function is refined. The analysis requires that the limiting ($n \to \infty$) observation operator be well defined. 
This is not the case in the above example, with increasing domain size and an increasing number of observations. 
The limit of an increasing domain size would be a domain of infinite size, which is difficult to describe with mathematical rigor.

\subsection{Example 2: Gaussian prior with squared Laplacian as precision matrix}
We now consider a Gaussian prior on $z\in[0,L]$
with mean zero and with a precision matrix derived from the squared Laplacian,
which leads to a pentadiagonal precision matrix after discretization.
As before, we fix the discretization and chose $\Delta z = 0.01$.
The Laplacian is approximated by the $n=L/\Delta z$ dimensional matrix
$\mathbf{L} = 1/(\Delta z)^2\mathbf{A}$,
where $\mathbf{A}$ is a matrix with $2$ on the main diagonal  
and $-1$ on the first upper and lower diagonal.
We define the prior precision matrix by
$\mathbf{\Omega}=(1/\rho^2\mathbf{I}+\mathbf{L})^2$,
where $\rho=0.06$.
As in example~1, we collect data by collecting noisy measurements
of a prior sample every $L/(2\Delta z)$ units.
The data are perturbed by a Gaussian random variable with mean
zero and covariance $\mathbf{R}=\mathbf{I}$.
As in Example~1, the discretization and observation network 
yields $100$ discrete state variables and $50$ data points per unit length.
The condition number of the prior precision or covariance matrix is about $21\cdot 10^3$.

In Figure~\ref{fig:Ex2},
we illustrate the banded structure of the prior and posterior covariance
and precision matrices for a problem with domain length $L=0.5$.
\begin{figure}[t]
\centering
\includegraphics[width=.4\textwidth]{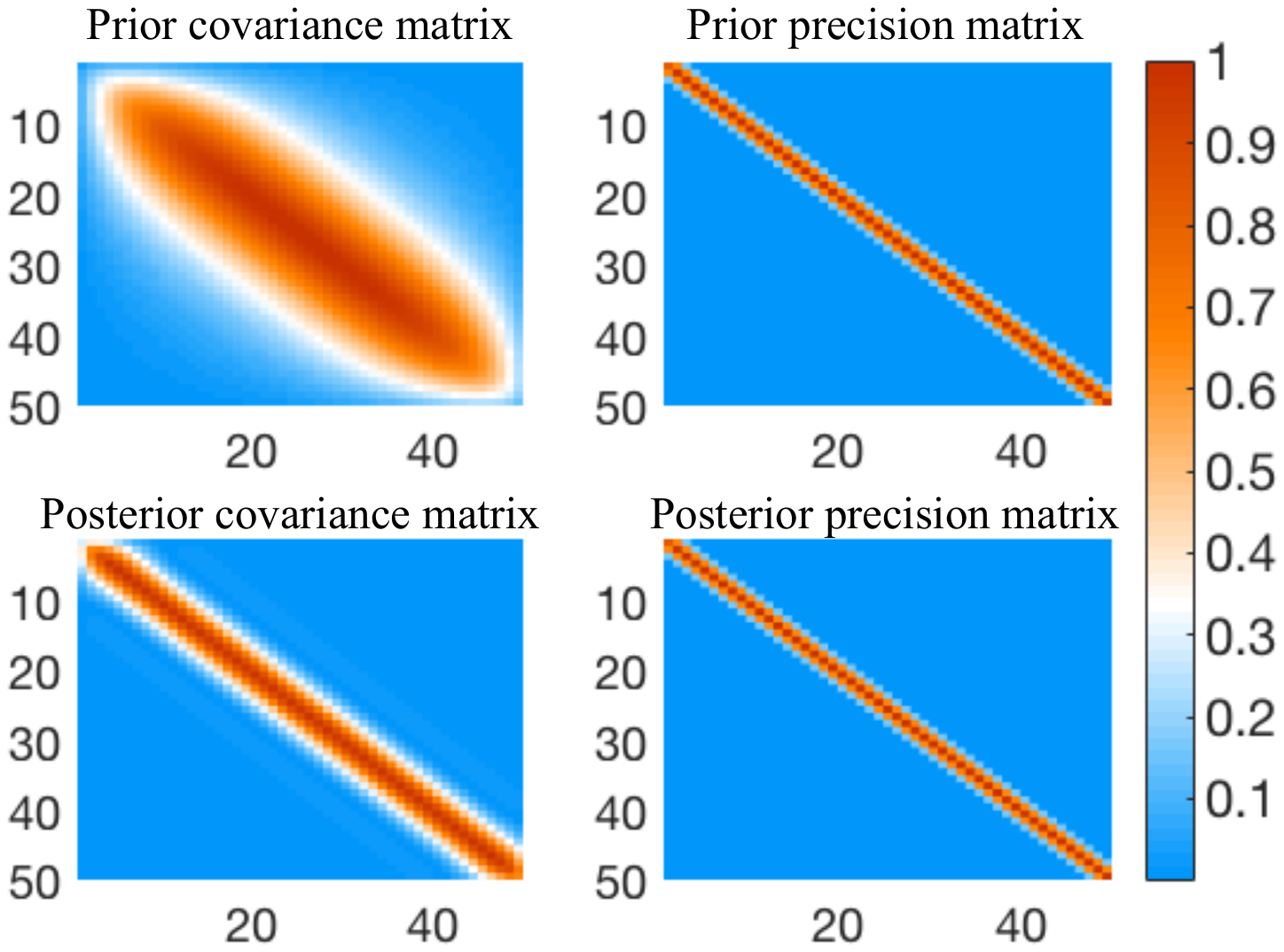}
\includegraphics[width=.4\textwidth]{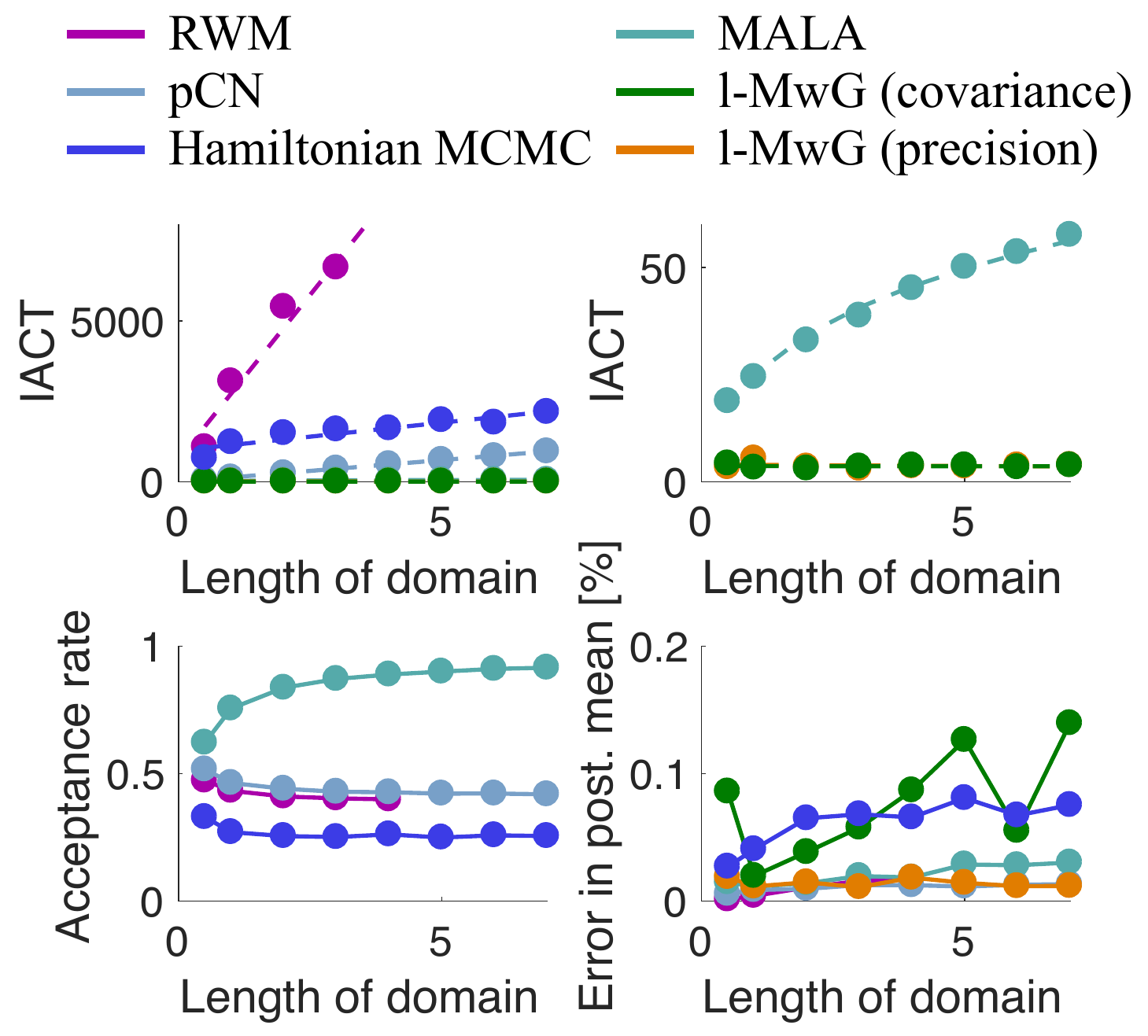}
\caption{
Covariance and precision matrices, IACT, acceptance rate
and errors for Example~2.
\textit{Left panels}. 
Prior and posterior covariance and precision matrices of Example~2 with $L=0.5$.
Top row: prior covariance (left) and precision matrix (right).
Bottom row: posterior covariance (left) and precision matrix (right).
\textit{Right panels}.
Top row: IACT as a function of domain length for
RWM, pCN, Hamiltonian MCMC (left)
and MALA and l-MwG (right).
Bottom row: acceptance rate as a function of domain length (left)
and mean squared error in posterior mean (right).
}
\label{fig:Ex2}
\end{figure}
As in example~1, the prior covariance matrix is nearly banded
(off-diagonal elements are small compared to diagonal elements),
while the prior precision matrix is banded (pentadiagonal).
Thus, as before, we localize the prior covariance matrix
by setting all elements below a threshold of 0.01 equal to zero.
The l-MwG sampler in prior precision matrix implementation
does not require localization since the prior precision is banded.
We also apply RMW, pCN, MALA, Hamiltonian MCMC,
and tune their step-size
as described in Example~1, with our findings summarized in Table~\ref{tab:Prior1}.
The number of samples we consider for each algorithm is also given in Table~\ref{tab:Prior1}.

We vary the domain length $L$ from $L=0.5$ to $L=7$,
apply the various samplers and compute the corresponding IACTs.
This leads to the scalings of IACT with domain length 
(or, equivalently, dimension) 
as shown in Table~\ref{tab:Prior1},
and as illustrated in Figure~\ref{fig:Ex2}.
While the scalings of IACT with domain size
for RWM, Hamiltonian MCMC and pCN are equal (all scale linearly),
we find that IACT of pCN and Hamiltonian MCMC is reduced compared to RWM.
As in Example~1, we find that 
the l-MwG samplers and MALA exhibit smaller IACT
than the other algorithms we tested,
and that IACT of the l-MwG samplers remains constant
as we increase the domain length.
We further find that the acceptance rate of MALA
first increases with $L$, but then levels off 
(due to our tuning of the step-size).
Similarly, the acceptance rate of RWM, pCN, and Hamiltonian MCMC
first decreases, but then levels off for large $L$.
All algorithms yield a small error (see equation~(\ref{eq:Error})),
which, in the case of l-MwG in covariance matrix implementation,
supports our claim that the localized problem is indeed
a small perturbation of the unlocalized problem. 
The reasons for increase in IACT with dimension for pCN
are the same as in example~1.

\subsection{Example 3: banded covariance, but full precision matrix}
We now consider a discrete Gaussian prior and Gaussian inverse problem that
does not necessarily have a limit as the discretization is refined.
We pick this perhaps unphysical problem to illustrate limitations of l-MwG
when the condition number is large and when it increases with dimension.

We pick a dimension, $n$, and chose a zero mean
and a covariance matrix given by an $n\times n$ matrix $\mathbf{C}$
which has 2 on the main diagonal and $-1$ on the first upper and lower diagonal.
We vary the dimension between $n=50$ and $n=700$
which leads to condition numbers between $1\cdot 10^3$ and $2\cdot 10^5$
(the condition number strictly increases with dimension).
Because of the large condition number,
the banded covariance does not imply that the precision matrix is also banded.
In fact, for this prior, the covariance matrix is tridiagonal,
but the precision matrix is full, as illustrated in Figure~\ref{fig:Prior3}.
\begin{figure}[tb]
\centering
\includegraphics[width=.4\textwidth]{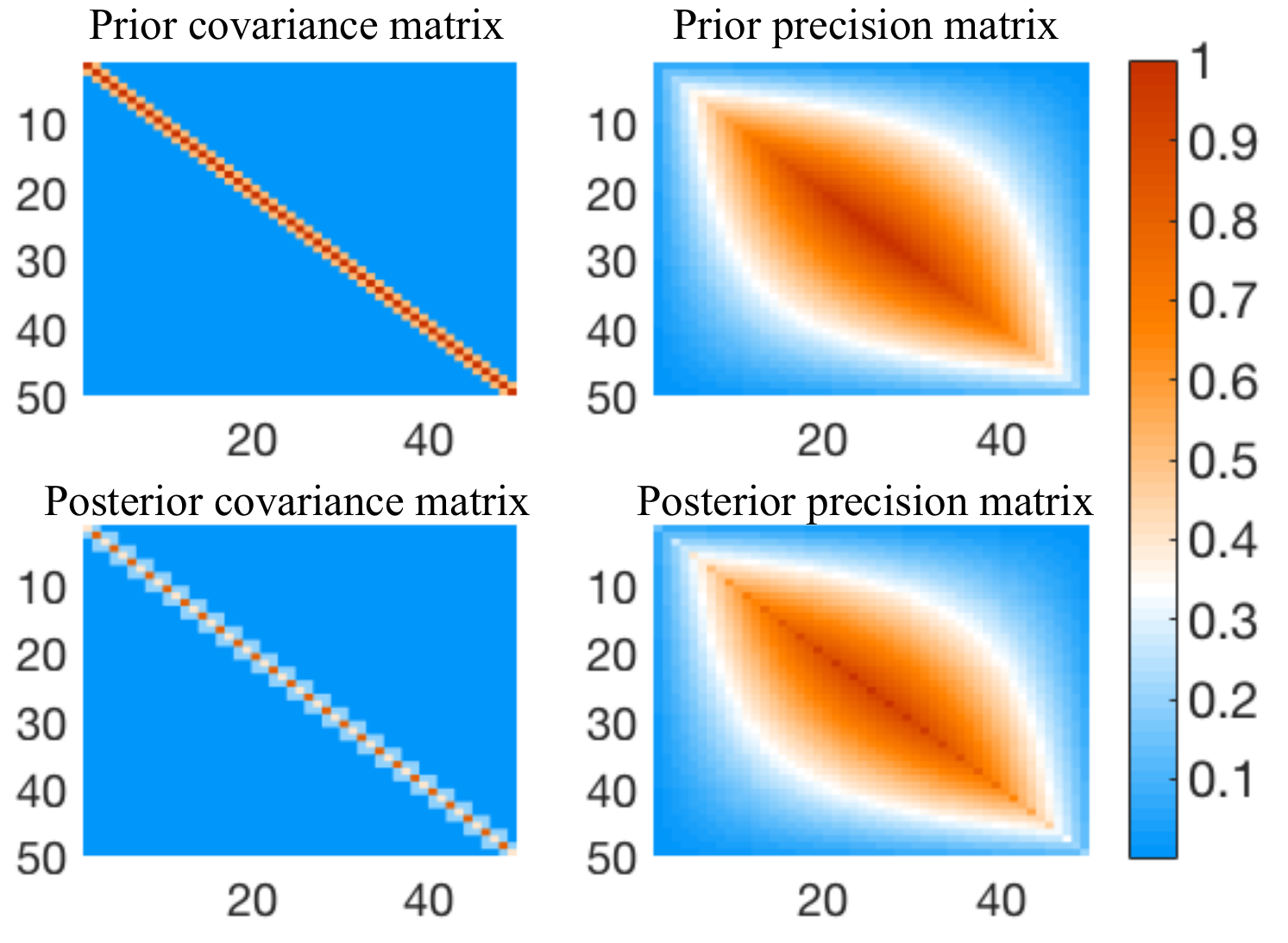}
\includegraphics[width=.4\textwidth]{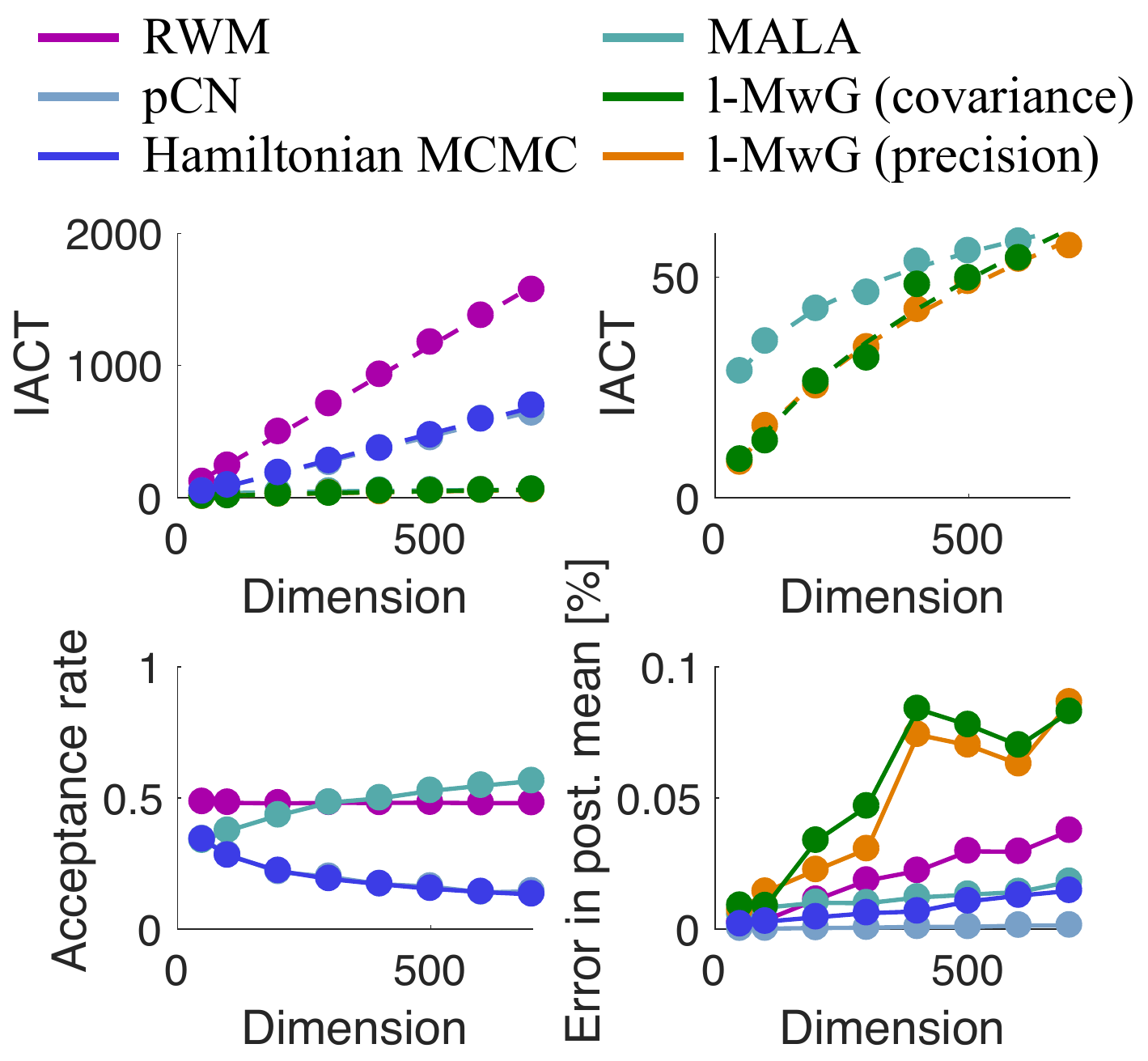}
\caption{
Covariance and precision matrices, IACT, acceptance rate
and errors for Example~3.
\textit{Left panels}. 
Prior and posterior covariance and precision matrices of Example~3 with $L=0.5$.
Top row: prior covariance (left) and precision matrix (right).
Bottom row: posterior covariance (left) and precision matrix (right).
\textit{Right panels}.
Top row: IACT as a function of domain length for
RWM, pCN, Hamiltonian MCMC (left)
and MALA and l-MwG (right).
Bottom row: acceptance rate as a function of domain length (left)
and mean squared error in posterior mean (right).
}
\label{fig:Prior3}
\end{figure}
As in examples~1 and~2,
we observe every other variable and perturb these measurements
by a Gaussian with mean zero and covariance $\mathbf{R}=\mathbf{I}$,
where $\mathbf{I}$ is the identity matrix of dimension $n/2$.
Because the prior precision matrix is not banded,
the posterior precision is also not banded as shown in Figure~\ref{fig:Prior3}.

We solve the Gaussian inverse problem 
by RWM, MALA, pCN, Hamiltonian MCMC and l-MwG.
The set-up and tuning of the algorithms are as described in Examples~1 and~2,
and summarized in Table~\ref{tab:Prior1}.
The l-MwG sampler in covariance matrix implementation does not require
localization since the prior covariance is banded (tridiagonal).
The precision matrix on the other hand cannot be localized efficiently
because the bandwidth of the prior precision matrix is large (see Figure~\ref{fig:Prior3}),
and also increases with $n$.
We thus do \textit{not} localize the l-MwG sampler in precision matrix implementation
and condition on all variables during sampling, 
not only on nearby variables.

For each algorithm and dimension, 
we compute IACT (see Table~\ref{tab:Prior1} for the rates with which
IACT increases with dimension for the various algorithms).
We observe that IACT increases for all algorithms we consider,
as illustrated by Figure~\ref{fig:Prior3}. 
We observe that IACT for the l-MwG samplers now also increases.
The rate is $n^{1/2}$ which is larger than the rate of MALA which,
as in examples~1 and~2, is equal to $n^{1/4}$.
The reasons for increase in IACT with dimension for pCN
are the same as in examples~1 and~2.
The increase in IACT of l-MwG with dimension can be understood within
our theory by the increase in condition number with dimension.
This numerical experiment thus suggests that our assumption
of moderate condition number is indeed necessary
(rather than being a tool for proving the theorems),
which corroborates our claim that l-MwG
is effective (nearly dimension independent) for a unified class of problems
with banded precision \emph{and} covariance matrix.
The bandedness of one implies bandedness of the other
by the small condition~number.

\subsection{Example 4: image deblurring}
We consider a linear inverse problem 
similar to an inverse problem in image deblurring
and assume that $\bfx$ is a column-stack of the pixels of a 2D image.
For example, for an image with $64\times 64$ pixels,
the dimension of $\bfx$ is $n=64^2=4,096$.
Thus, for high-resolution image deblurring (images with a large number of pixels),
the dimension of $\bfx$ is huge ($10^7$ for a $4,000\times 4,000$ image).
We consider the images of sizes $32\times32$, $64\times64$, $128\times128$, $256\times256$
and shown in in Figure~\ref{fig:Ex4}.
\begin{figure}[tb]
\centering
\includegraphics[width=.8\textwidth]{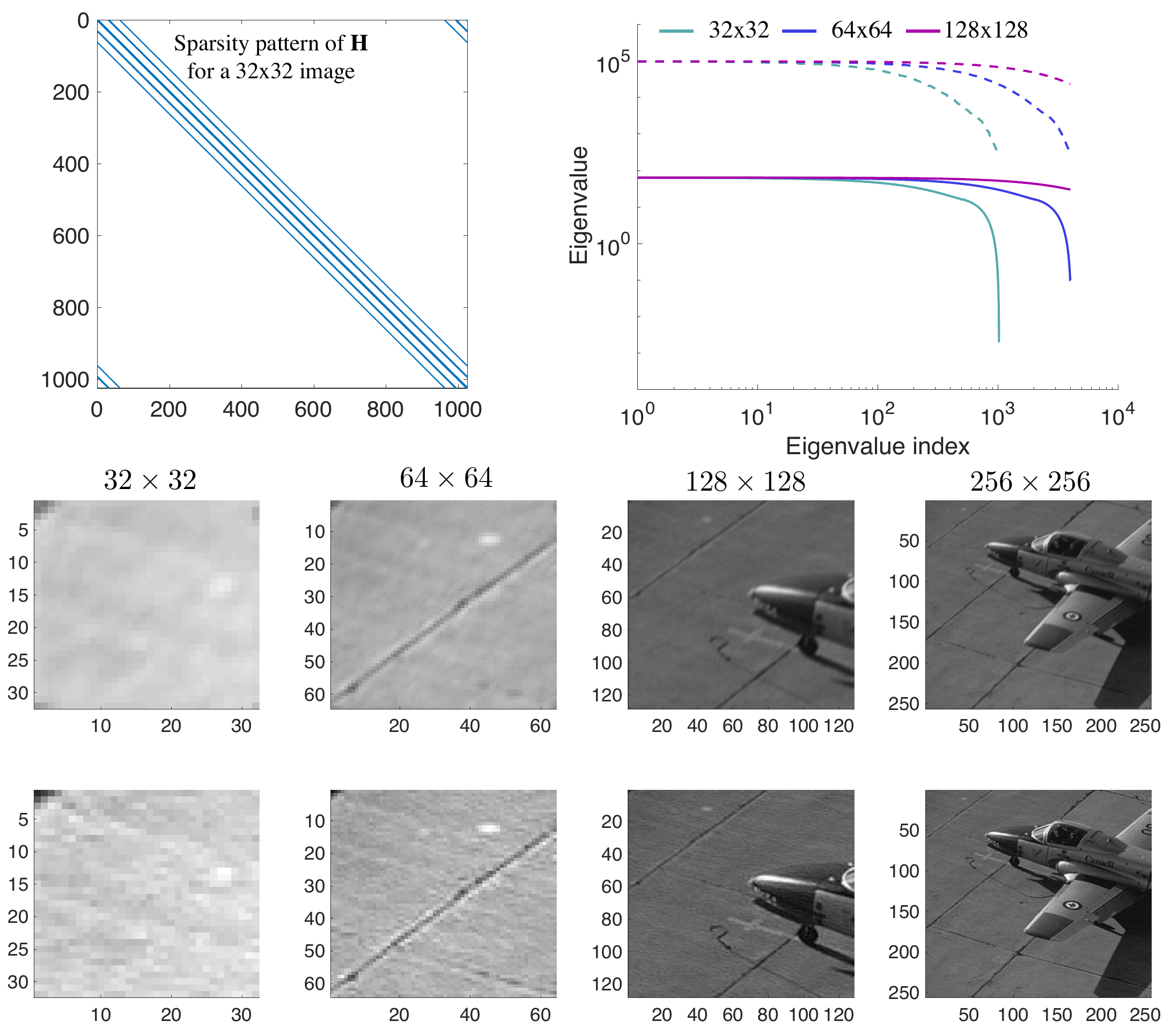}
\caption{
Top left: sparsity pattern of the matrix $\bfH$ for a $32\times 32$ image.
Top right: the largest eigenvalues of the prior precision (solid)
and posterior precision (dashed) matrices of $32\times 32$ (turquoise),
$64\times 64$ (blue) and $128\times 128$ (purple) images.
Center row: blurred images.
Bottom row: mean of $10^4$ samples of the Gibbs sampler.
}
\label{fig:Ex4}
\end{figure}
These images are obtained by cutting square images out of a given, slightly larger image.
We use this example to illustrate that IACT of the Gibbs sampler can remain constant 
as the dimension of $\bfx$ and the number of observations increase with image size.
Many practical difficulties, e.g., 
estimation of noise and regularization parameters (see, e.g., \cite{FN16,Bardsley12,FHLMW16}),
are neglected in this example.

We assume that the image $\bfx$ is blurred by multiplication from the left with a
given $n\times n$ matrix $\bfH$ and that a noisy version of the
blurred image, $\bfy$, is available.
Thus,
\begin{equation*}
	\bfy=\bfH\bfx + \bf{\eta}, \quad \bf{\eta}\sim\mathcal{N}(0,\lambda^{-1}\bfI),
\end{equation*}
where $\lambda = 10^5$ describes a (small) Gaussian measurement noise.
The matrix $\bfH$ represents blurring the image by a Gaussian kernel
with standard deviation 0.7, which is ``realistic'' because 
it effectively averages only a few pixels in each direction.
Throughout this example we assume periodic boundary conditions.
The prior is a Gaussian with mean zero and the prior precision is the 
2D Laplacian, 
i.e., $p_0(\bfx)=\mathcal{N}(0,\delta^{-1} \bfL^{-1})$, where $\delta=10$.
For this set-up, the posterior distribution is the Gaussian 
\begin{equation*}
	p(\bfx\vert \bfb)\propto \exp\left(
	-\frac{\lambda}{2}\vert\vert \bfH\bfx -\bfb\vert\vert^2
	-\frac{\delta}{2}\vert\vert \bfL^{1/2}\bfx\vert\vert^2
	\right),
\end{equation*}
and the posterior precision matrix is given by
\begin{equation*}
	\bfOmega = \lambda \bfH^T\bfH + \delta \bfL.
\end{equation*}

The prior precision matrix $\bfL$ is banded
and, thus, no localization of the prior is required.
In this example, localization of the observation matrix is straightforward:
since blurring occurs locally, 
only a few bands near the diagonal of $\bfH$ 
are ``significant'' and all other elements of $\bfH$ are ``small.''
We can thus localize $\bfH$ by setting all elements below a threshold equal to zero.
The threshold is 1\% of the maximum value of all elements of $\bfH$.
This gives rise to a sparse matrix $\bfH_\text{loc}$
whose sparsity pattern is illustrated for a $32\times 32$ image in Figure~\ref{fig:Ex4}.
Since $\bfH_\tloc$ and $\bf{L}$ are sparse,
the posterior precision matrix $\bfOmega_\tloc$ is also sparse (see also \cite{BL12}).
Errors due to localization of $\bfH$ are small.
We compute that $\vert\vert\bfH-\bfH_\tloc\vert\vert\approx 0.02$
and that the differences in the posterior covariance
are $\vert\vert\hat{\bfC}-\hat{\bfC}_\tloc\vert\vert\approx 3.3\cdot 10^{-3}$,
the error in the posterior mean is $\vert\vert\hat{\bfm}-\hat{\bfm}_\tloc\vert\vert\approx 1$
(independently of image size).

We can compute the effective dimension $n_{\text{eff},2}$ in \eqref{eqn:effdim2}
and the largest (at most $4000$) eigenvalues of the prior and posterior precision matrices
for some of the localized problems. 
Our results are shown in Figure~\ref{fig:Ex4} and displayed in Table~\ref{tab:ResultsImageDeblurr}.
\begin{table}[ht]
\begin{center}
\caption{Dimension, effective dimension and IACT of a Gibbs sampler
as a function of the image size.}
\label{tab:ResultsImageDeblurr}
\begin{tabular}{rcccc}
Image size: & $32\times32$ &$64\times64$ &$128\times128$ &$256\times256$ \\
\hline
Dimension: & 1,024 & 4,096 & 16,384 & 65,536\\
$n_{\text{eff},2}:$ & $4.8\cdot 10^8$ & $7.4\cdot 10^9$ & $1.2\cdot 10^{11}$& - \\
IACT (l-MwG): & 2.92 & 2.97 & 1.74 & 1.11
\end{tabular}
\end{center}
\end{table}
We note that the effective dimension and the number of non-negligible eigenvalues increase
with image size and dimension.
Thus, as in the previous examples, the high-dimension of this local inverse problem
is not caused by a large apparent, but small effective dimension,
or by an increasing number of negligible eigenvalues,
but rather \textit{all} dimensions of this problem are large.
We did not compute the effective dimension or eigenvalues
corresponding to the $256\times 256$ image
because the required matrix operations
were difficult to do (on our laptop) even when exploiting the sparsity of the matrices.

We implement a Gibbs sampler for the posterior distribution
using the posterior precision matrix $\bfOmega_\tloc$.
We define blocks using the 2D structure of the image rather than
consecutive elements of the column stacks. 
In specific, the component index is two dimensional, i.e. 
$[\bfx]_{i,j}$ denotes the $(i,j)$-th pixel of an $n'\times n'$ image. And the $(a,b)$-th block of size $q=q'\times q'$consists of entries
\[
\bfx_{a,b}=([\bfx]_{a+i,b+j},i,j=1,\cdots,q').
\]
Two blocks $\bfx_{a_1,b_1}$ and $\bfx_{a_2,b_2}$ are neighbors if $|a_1-a_2|\leq 1$ and $|b_1-b_2|\leq 1$. 
During one iteration of the Gibbs sampler,
we sample one block conditioned on the neighboring blocks.
The blocks are of size $16\times 16$ for the $32\times 32$ and $64\times 64$ images,
for the $128\times 128$ image we use blocks of size $32\times32$
and for the $256\times 256$ image we use blocks of size $64\times 64$.
For each image, we draw 10,000 samples from the posterior distribution using the Gibbs sampler.
We then compute IACT for every 8$th$ pixel.
The averages of these IACTs are displayed in Table~\ref{tab:ResultsImageDeblurr}.
IACT is near one for all four images we consider
and we emphasize that, as in previous examples,
IACT does not increase with dimension.
The slight decrease in IACT as dimension increases can be
attributed to the larger block size we use.
For example, if the block size is the size of the image,
we draw independent samples and IACT equals one.

We checked the results we obtained by the Gibbs sampler as follows.
The sample average is, to three digits, the same as the posterior mean
of the localized problem (as computed by linear algebra, rather than sampling).
We also compute the trace of the sample covariance matrix
and compared it to the trace of the sample covariance
of 10,000 samples obtained by a ``direct'' sampler
that makes use of the sparse Cholesky factorization of the posterior precision matrix.
The average of these variances agrees to the first three digits.
This is perhaps not surprising in view of the small IACT.

We also applied pCN and MALA for this problem,
but could not obtain useable results even for the smallest image (dimension $n=1,024$).
We tuned the step sizes of these algorithms,
but for all choices we considered,
we either accepted often, in which case the moves were too small (large IACT),
or we observed that larger moves were almost never accepted.
The reason for the difficulties with MALA and pCN
is the large effective dimension of this problem.

\subsection{Example 5: a nonlinear inverse problem}
We consider a nonlinear inverse problem whose likelihood 
involves numerical solution of the Lorenz'96 (L96) model \cite{L96}
\begin{equation*}
	\frac{\text{d}x_i}{\text{d}t} = (x_{i+1}-x_{i-2})x_{i-1} - x_i+8,
\end{equation*}
where $i=1,\dots,n$ and 
$x_{-1}=x_{n-1}$, $x_0=x_n$, $x_{n+1}=x_1$.
We consider L96 models with $n=40$ and $n=400$.
Specifically, we try to estimate the initial condition $\bfx_0=(x_1(0),\dots,x_n(0))^T$
given noisy observations of every other state variable at time $T=0.2$.
We write this as
\begin{equation}
\label{eq:L96Likelihood}
	\bfy = \bfH\mathcal{M}_{0\to T}(\bfx_0) + \mathbf{\eta},\quad \mathbf{\eta}\sim\mathcal{N}(\mathbf{0},\bfI),
\end{equation}
where  $\bfH$ is a $n/2\times n$ matrix
which selects every other component of $\bfx_T=\mathcal{M}_{0\to T}(\bfx_0)$;
$\mathcal{M}_{0\to T}$ is the numerical solver of the L96 model
from $t=0$ to $t=T$, i.e., $\mathcal{M}_{0\to T}(\bfx_0)$ 
maps the initial condition $\bfx_0$ to $\bfx_T$.
We use a fourth order Runge-Kutta  scheme with time step $\Delta t=0.01$.

We assume a Gaussian prior whose mean and covariance we obtain as follows.
We initialize the L96 model with an arbitrary initial condition  near the attractor
and perform a simulation for 100 time units ($10^4$ time steps).
The mean of this L96 trajectory is used as the prior mean,
and we use a localized version of the covariance matrix computed from the trajectory
as the prior covariance matrix. 
The localization is done as follows.
We first compute the Hadarmard (element by element) product of
the sample covariance and a localization matrix whose elements
are $[\bfA]_{i,j}=\exp\left(-\left(d(i,j)/(3\,\sqrt{2})\right)^2\right)$,
where $d$ is a periodic distance: $d(i,j)=\min\{|i-j|,|i-j+n|, |i-j-n|\}$.
Subsequently, off-diagonal elements below a threshold
(1\% of the largest element) are set to zero. 
This results in a localized, sparse prior covariance matrix.
The eigenvalues of the prior covariances for problems of dimensions
$n=40$ and $n=400$ are shown in Figure~\ref{fig:Ex5Spectra}.
\begin{figure}[tb]
\centering
\includegraphics[width=.5\textwidth]{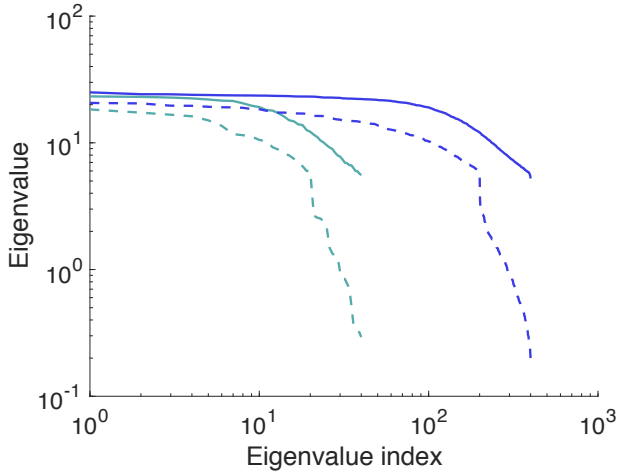}
\caption{
Eigenvalues of the prior covariance matrices for the L96 model.
Turquoise: $n=40$.
Blue: $n=400$.
Solid: prior covariance matrix.
Dashed: approximate posterior covariance.
}
\label{fig:Ex5Spectra}
\end{figure}
As in the above examples, the eigenvalues do not decay quickly
and the number of non-negligible eigenvalues increases with dimension.

The localized prior and a likelihood defined by equation~(\ref{eq:L96Likelihood})
define the (non-Gaussian) posterior distribution
\begin{equation}
\label{eq:Ex5Post}
	p(\bfx_0\vert\bfy) \propto p_0(\bfx_0)
	\exp\left(-\frac{1}{2} \vert\vert\bfH\mathcal{M}_{0\to T}(\bfx_0)-\bfy\vert\vert^2\right).
\end{equation}
We minimize $-\log(p(\bfx\vert\bfy))$ by Gauss-Newton,
using adjoints for gradient calculations, and use the optimization result
as a reference solution for the MCMC results. 
We compute the inverse of the Gauss-Newton approximation of the Hessian
of $-\log(p(\bfx\vert\bfy))$, evaluated at its minimizer, 
and use this matrix, $\bfP$, as an approximation of posterior covariance
(see, e.g., \cite{ALB16,TalagrandCourtier,sakov2012,TLM} for discussion
of these approximations in NWP).
The eigenvalues of the approximate posterior covariance matrices
for problems of dimension $n=40$ and $n=400$
are shown in Figure~\ref{fig:Ex5Spectra}.
Note that the number of non-negligible eigenvalues increases with dimension.
With the approximate posterior covariance $\bfP$ we can also compute an effective dimension. 
Here we use $n_{\text{eff},1}$  as in \eqref{eqn:effdim1},
because the problem is naturally formulated in terms of covariance matrices, 
rather than precision matrices.
We compute $n_{\text{eff},1}=17.94$ when $n=40$
and $n_{\text{eff},1}=181.36$ when $n=400$.
As in the above examples, this effective dimension increases with dimension.

We use MALA, pCN and l-MwG to draw samples from the posterior distribution~(\ref{eq:Ex5Post}),
which is not  Gaussian because the L96 model is nonlinear,
i.e., $\mathcal{M}_{0\to T}(\bfx_0)$ is not a linear function of $\bfx_0$.
The MALA proposal we use is 
\begin{equation*}
	\bfx_{k+1}' = \bfx_k+\frac{\sigma^2}{2}\nabla \log p(\bfx_k)+\sigma\, \xi
\end{equation*}
where $\xi\sim\mathcal{N}(0,\bfC_\text{MALA})$.
We chose the covariance matrix $\bfC_\text{MALA}$
to be a diagonal $n\times n$ matrix 
whose diagonal elements are the diagonal elements of $\bfP$.
For this example, this choice for the covariance of $\xi$ normalizes each dimension, and
it gives better results, in terms of  acceptance ratio and IACT, than using the identity matrix.
We consider several choices for the step size $\sigma$,
and for each choice generate $10^5$ samples.
For fixed $\sigma$, we compute IACT for all $n$ variables,
and then average.
We also compute the average (along the chain) of the acceptance ratio.
The results for L96 models of dimensions $n=40$ and $n=400$ are shown in Figure~\ref{fig:Ex5}.
\begin{figure}[tb]
\centering
\includegraphics[width=.7\textwidth]{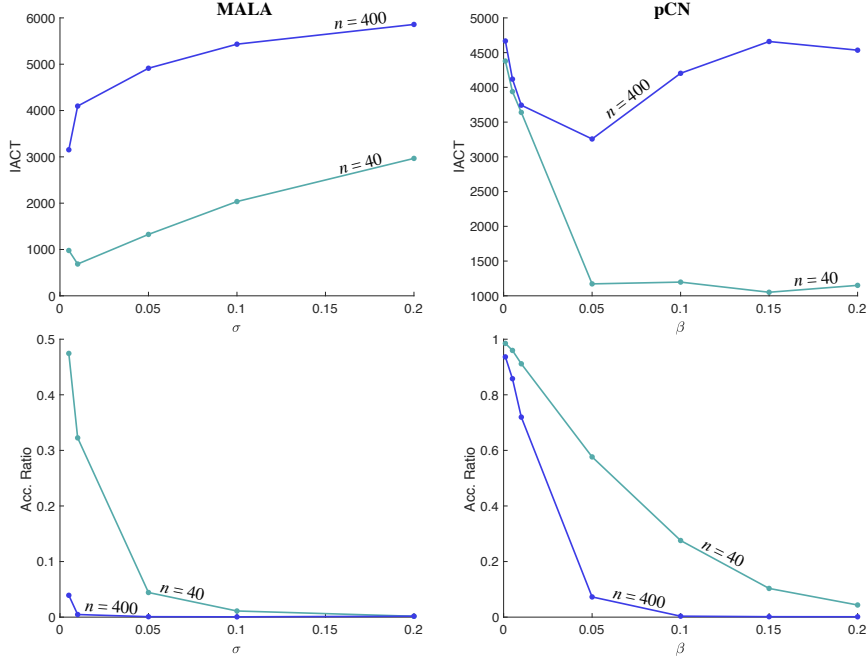}
\caption{
Average IACT (top) and average acceptance ratio (bottom)
as a function of the step-size.
Left: MALA. Right: pCN.
}
\label{fig:Ex5}
\end{figure}
We note that a small step-size is required in order
to retain a reasonable acceptance ratio.
As we increase the dimension from $n=40$ to $n=400$,
we note that the acceptance ratio for a given step size $\sigma$ drops.
For this  reason, IACT increases as the dimension increases.
The smallest (average) IACTs we could achieve for the $n=40$ and $n=400$ dimensional
problems are shown in Table~\ref{tab:Ex5IACT}.
It is evident that IACT increases with dimension in this problem. 

\begin{table}[htp]
\begin{center}
\begin{tabular}{rcccccc}
& &MALA &pCN &l-MwG-2 & l-MwG-4 &l-MwG-8 \\
\hline
&$n=40$ & 686 & 1,051 &55&60&266\\
&$n=400$ &3,153  & 3,257 &43&81&257
\end{tabular}
\caption{IACT of MCMC methods: The IACT of MALA and pCN are minimized with $\sigma$ and $\beta$. The q in l-MwG-q  is the block size.  }
\label{tab:Ex5IACT}
\end{center}
\end{table}%

The pCN proposal we use is 
\begin{equation*}
	\Delta\bfx_{k+1}' = \sqrt{1-\beta^2}\Delta\bfx_k+\beta\, \xi_k
\end{equation*}
where $\xi_k\sim\mathcal{N}(0,\bfC)$ and where $\Delta\bfx$
are perturbations from the prior mean.
We consider several choices for $\beta$ and, for each choice,
we generate $10^5$ samples and compute the average IACT
and average acceptance ratio as above.
Our results are shown in the right panels of Figure~\ref{fig:Ex5}.
We find that the acceptance ratio of pCN is generally higher than that of MALA,
but IACT of pCN and MALA are comparable.
We also observe that IACT of pCN increases with dimension.
The smallest IACT we found are shown in comparison with those of MALA
in Table~\ref{tab:Ex5IACT}.

We implement the l-MwG sampler using the prior covariance matrix and observation localization.
The blocks are consecutive components of $\bfx_0$
(but taking into account the periodicity of the L96 model)
and the block sizes we tested are $2$, $4$ and $8$.
Localization of the observation function is done by 
using, for each block, all observations within the block 
as well as the neighboring two observations on each side
(accounting for the periodicity of L96).
For each block size, we generate $10^4$ samples 
and compute an average IACT as above.
Results are shown in Table~\ref{tab:Ex5IACT}.
Note that IACT of l-MwG is significantly smaller
than the IACT we computed for pCN or MALA.
More importantly we observe that IACT does not 
increase significantly when we increase the dimension.

We illustrate the localized prior and posterior distributions
in Figure~\ref{fig:Ex5Samples}.
\begin{figure}[tb]
\centering
\includegraphics[width=.7\textwidth]{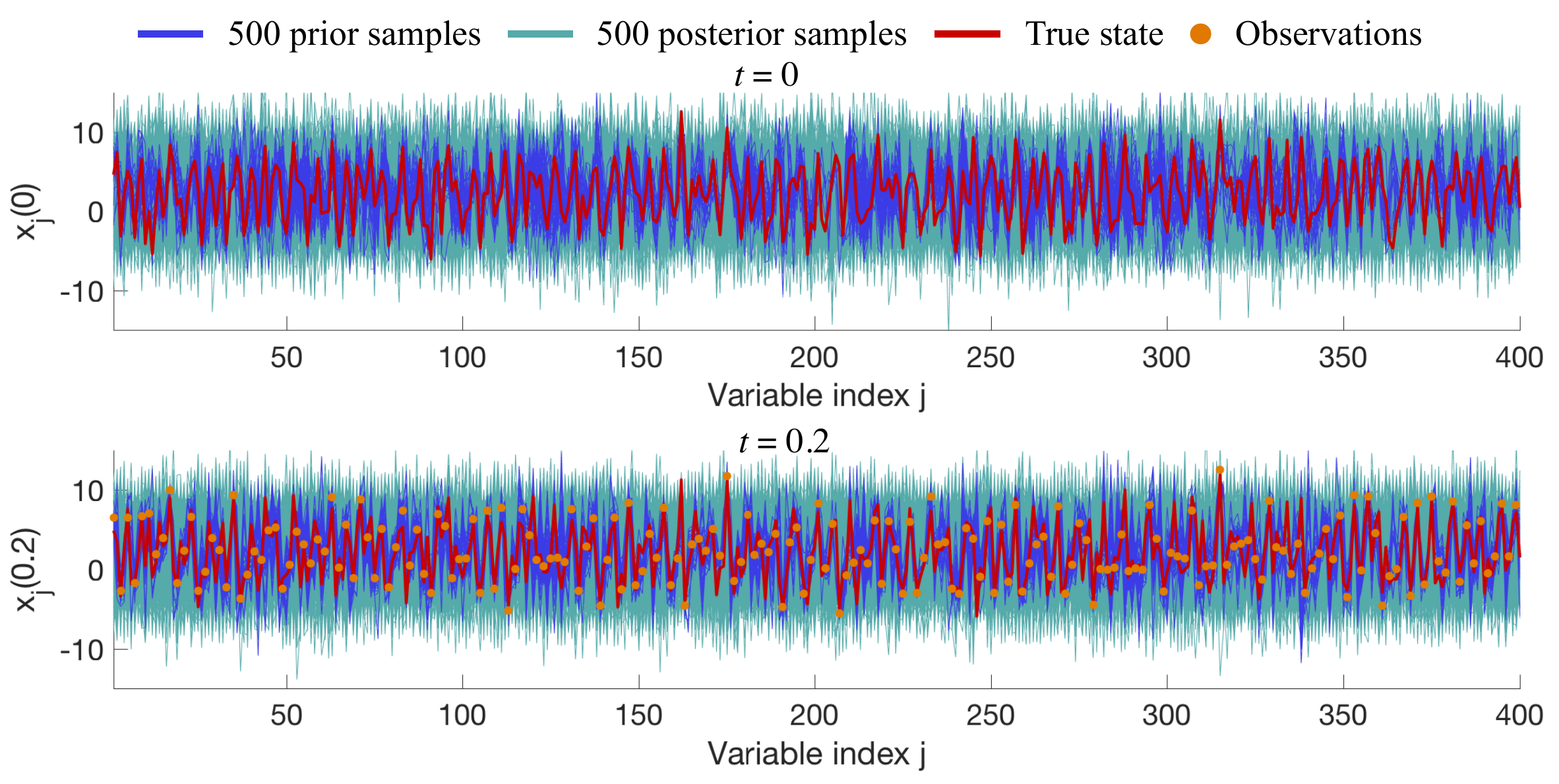}
\caption{
Top: 500 prior samples (turquoise), 
the true state (red) and 500 posterior samples
obtained by l-MwG (blue).
Bottom:
500 model states at time $t=0.2$
obtained by using the L96 model on the prior (turquoise) and posterior (blue) samples.
Shown in red is the true state and the orange dots are the observations. 
}
\label{fig:Ex5Samples}
\end{figure}
Samples of the prior distribution serve as a reference
and 500 samples of the prior are shown in turquoise;
500 samples of the posterior distribution,
obtained by l-MwG with block size eight,
are shown in blue.
The true state is shown in red.
Localization of the prior and using l-MwG to draw samples from the
localized problem produces a meaningful posterior distribution: 
the uncertainty is reduced and the posterior samples are centered around the true state.
The lower panel of Figure~\ref{fig:Ex5Samples} illustrates
how the uncertainty in the localized prior and posterior distributions propagates through time.
We show the 500 samples of the prior (turquoise) and posterior (blue)
mapped to time $T=0.2$ by the L96 dynamics $\mathcal{M}_{0\to T}(\bfx_0)$.
Again, the true state (at time $T$) is shown in red and lies well within the cloud of posterior samples.
Plotting states at time $T$ also allows us to compare to the observations,
shown as orange dots, and we see good agreement between
the model states at time $T$,  generated from posterior samples, 
and the observations.

We emphasize that l-MwG, or in fact any MCMC method,
may not be appropriate for solving this inverse problem.
The minimizer of $-\log(p(\bfx_0\vert\bfy))$ 
and the associated Hessian-based covariance are good approximations
of the posterior mean and covariances we computed by MCMC,
but Gauss-Newton optimization is more efficient than MCMC for this problem.
Our main messages for this example are that
(\textit{i}) l-MwG can be used on nonlinear problems;
(\textit{ii}) l-MwG, or other MCMC samplers that make use
of the local structure of the problem, may be more effective
than MCMC methods that do not make use of local problem structure (pCN and MALA);
and (\textit{iii}) increasing the dimension of this problem 
has almost no effect on the performance of l-MwG,
but MCMC samplers that do not make use of local problem structure (pCN and MALA)
are ineffective if the dimension and effective dimension are large.
  
\section{Summary and conclusions}
\label{sec:Summary}
The main goal of this paper is to demonstrate
that ideas of localization in numerical weather prediction (NWP) and
the ensemble Kalman filter (EnKF) have relevance in inverse problems and MCMC.
During localization, one restricts prior statistical interactions
(correlations and\slash or conditional dependencies) and the effects of
observations to neighborhoods.
We have discussed conditions under which such ideas can be used
in the context of inverse problems and their solution by MCMC.
For example, we proved that localization introduces small errors
for a class of linear ``local'' problems,
but expect that this also holds for nonlinear local problems.

We reviewed the Gibbs sampler as a natural MCMC
sampler that exploits local problem structure.
We observed that its performance is dimension independent when
sampling Gaussian distributions with banded precision and covariance matrices.
We presented a Metropolis-within-Gibbs sampler
that can be applied to linear or nonlinear problems.
We demonstrated our ideas in several numerical examples in which 
Gibbs samplers outperformed MALA, Hamiltonian MCMC, RWM and pCN.

Localization is useful for inverse problems only 
if there are interesting applications in which localization can be applied.
We speculate that local structure is common in physics and engineering,
with NWP being the most spectacular example.
Finally, we have discussed that the notion of high dimensionality
in local problems is different from what 
is usually assumed in function space MCMC.
The literature on function space MCMC focuses
on problems with a large apparent dimension,
but few observations and low-rank updates from prior to posterior.
Localization is useful in a different scenario,
in which the apparent and effective dimensions,
and the number of observations, are large,
and updates from prior to posterior are not low-rank.

Our study neglects many practical challenges.
For example, localization may require some tuning,
which in itself can be computationally expensive.
More importantly,
we have not rigorously defined localization for non-Gaussian problems, 
where enforcing bandedness of the covariance matrices may not be sufficient
because higher moments become important.
Taking inspiration instead from the sparsity of the
precision matrix, a useful route may be to consider the conditional independence
structure of more general non-Gaussian Markov random fields \cite{spantini2017}.
We hope to address these issues in future work.

\section*{Acknowledgements}
M.~Morzfeld gratefully acknowledges support by the Office of Naval Research 
(grant number N00173-17-2-C003), by the National Science Foundation (grant DMS-1619630), 
and by the Alfred P. Sloan Foundation (Sloan Research Fellowship). 
X.T.~Tong gratefully acknowledges support by the National University of Singapore (grant R-146-000-226-133).
Y.M.~Marzouk gratefully acknowledges support from the DOE Office of Advanced Scientific Computing Research (grant DE-SC0009297).

\appendix
\section{Bias caused by localization} 
In Section~\ref{sec:Localization}, 
we propose to localize the prior covariance or precision matrix, and apply the l-MwG. 
This leads to samples from the localized posterior distribution,
not from the exact posterior distribution.
In this section we show that the difference between these two distributions can be small. 

First we need the following estimate.
\begin{lem}
\label{lem:perturbinverse}
Suppose $\bfA$ is a symmetric positive definite matrix, and $\bfB$ is a symmetric matrix such that $\|\bfB\|\leq \|\bfA^{-1}\|^{-1}$, then 
\[
\|(\bfA+\bfB)^{-1}-\bfA^{-1}\|\leq \frac{\|\bfA^{-1}\|^2\|\bfB\|}{1-\|\bfB\|\|\bfA^{-1}\|}.
\]
\end{lem}
\begin{proof}
Let $\bfB=\Psi L \Psi^T$ be the eigenvalue decomposition of matrix $\bfB$. Remove columns of $\Psi$ that are eigenvectors of value $0$, so $L$ only contains nonsingular terms. Then by the Woodbury matrix identity
\[
\bfA^{-1}-(\bfA+\Psi L \Psi^T)^{-1}=\bfA^{-1} \Psi(L^{-1}+\Psi^T \bfA^{-1} \Psi)^{-1} \Psi^T \bfA^{-1}.
\]
Let $v$ be the eigenvector corresponds to the largest absolute eigenvalue of $(L^{-1}+\Psi^T \bfA^{-1} \Psi)^{-1}$. Then
\[
|v^T (L^{-1}+\Psi^T \bfA^{-1}\Psi ) v|\geq  |v^TL^{-1} v|-|v^T\Psi^T \bfA^{-1} \Psi v|\geq \|L\|^{-1} -\|\bfA^{-1}\|. 
\]
Moreover, $ \|L\|=\|\bfB\|$, so $\|(L^{-1}+\Psi^T \bfA^{-1} \Psi)^{-1}\|\leq  (\|\bfB\|^{-1} -\|\bfA^{-1}\|)^{-1}$, and
\[
\|(\bfA+\bfB)^{-1}-\bfA^{-1}\|\leq \frac{\|\bfA^{-1}\|^2}{ \|\bfB\|^{-1}-\|\bfA^{-1}\|}.
\]
\end{proof}

\begin{prop}
\label{prop:loc}
Let $\bfC$ be a prior covariance matrix
and $\bfH$ be an observation matrix.
Let $\bfC_\tloc$ and $\bfH_\tloc$
be localized covariance and observation matrices,
and let  $\delta_C$ and $\delta_H$ be as defined in \eqref{eq:deltaC} and \eqref{eq:deltaH}.
\begin{enumerate}[a)]
\item 
The localized prior covariance and observation matrices
are small perturbations of the unlocalized covariance 
and observation matrices in the sense that
\[
\|\bfC-\bfC_\tloc\|\leq \delta_C, \quad \|\bfH_\tloc-\bfH\|\leq \delta_H.
\]
Moreover, if $\delta_C\leq \|\bfC^{-1}\|^{-1}$, then $\bfC_\tloc$ is positive semidefinite. 
\item 
If $\delta_C\leq \|\bfC^{-1}\|^{-1},\Delta_1\leq \|\bfChat\|^{-1}$,
the localization creates a small perturbation of the posterior covariance matrix
in the sense that
\[
\|\bfChat-\bfChat_\tloc\|\leq \frac{\|\bfChat\|^2\Delta_1}{1-\Delta_1 \|\bfChat\|},\quad \Delta_1:=\frac{\|\bfC^{-1}\|^2\delta_C }{1-\delta_C \|\bfC^{-1}\|}+(2\delta_H+\delta_H^2) \|\bfR^{-1}\|.
\]
\item 
Under the same conditions as in b),
the localization creates a small perturbation of the 
posterior mean in the sense that
\begin{align*}
\|\bfmhat_\tloc-\bfmhat\|&\leq \left(\frac{\|\bfC^{-1}\|^2 \|\bfChat\|\delta_C }{(1-\delta_C \|\bfC^{-1}\|)(1-\Delta_1 \|\bfChat\|)}+\frac{\|\bfC^{-1}\|\|\bfChat\|^2\Delta_1}{1-\Delta_1 \|\bfChat\|}\right) \|\bfm\|\\
&\quad\quad+\frac{\|\bfChat\|\|\bfR^{-1}\|}{1-\Delta_1 \|\bfChat\|}(\|\bfChat\|\Delta_1+\delta_H ) \| \|\bfy\|. 
\end{align*}
\end{enumerate}
\end{prop}

\begin{proof}
For  claim a), Define $\Delta=\bfC-\bfC_\tloc,$ apply Lemma \ref{lem:norm} with block size $1$, 
\[
\|\bfC-\bfC_\tloc\|\leq \left(\max_i \sum_{j}|[\Delta]_{i,j}| \right)=\delta_C.
\]
The bound $\|\bfH-\bfH_\tloc\|\leq \delta_H$ follows the same argument. 

For  claim b), we will exploit the identity $\bfChat^{-1}=\bfC^{-1}+\bfH^T \bfR^{-1} \bfH$. By Lemma \ref{lem:perturbinverse}, 
\[
\|\bfC^{-1}-\bfC_\tloc^{-1}\|\leq \frac{\|\bfC^{-1}\|^2\delta_C }{1-\delta_C \|\bfC^{-1}\|}. 
\]
Under our normalization assumption that $\|\bfH\|=1$, 
\begin{align}
\notag
\|\bfH^T \bfR^{-1} \bfH-\bfH^T_\tloc \bfR^{-1} \bfH_\tloc\|
&\leq  \|(\bfH^T-\bfH^T_\tloc)\bfR^{-1}\bfH\|+\|\bfH^T\bfR^{-1}(\bfH-\bfH_\tloc)\|\\
\notag
&\quad\quad+\|(\bfH^T-\bfH^T_\tloc)\bfR^{-1}(\bfH-\bfH_\tloc)\|\\
\label{tmp:HRH}
&\leq 2\delta_H \|\bfR^{-1} \bfH\|+ \delta_H^2 \|\bfR^{-1}\|\leq (2\delta_H+\delta_H^2) \|\bfR^{-1}\|. 
\end{align}
Therefore
\[
\|(\bfC^{-1}+\bfH^T \bfR^{-1} \bfH)-(\bfC_\tloc^{-1}+\bfH^T_\tloc \bfR^{-1} \bfH_\tloc)\|\leq \frac{\|\bfC^{-1}\|^2\delta_C }{1-\delta_C \|\bfC^{-1}\|}+(2\delta_H+\delta_H^2) \|\bfR^{-1}\|=:\Delta_1.
\]
We apply Lemma \ref{lem:perturbinverse} to $\bfChat=(\bfC^{-1}+\bfH^T \bfR^{-1} \bfH)^{-1}$, and obtain
\[
\|\bfChat-\bfChat_\tloc\|\leq \frac{\|\bfChat\|^2\Delta_1}{1-\Delta_1 \|\bfChat\|}.
\]
As a consequence, we also have that $\|\bfChat_\tloc\|\leq \frac{\|\bfChat\|^2\Delta_1}{1-\Delta_1 \|\bfChat\|}+\|\bfChat\|= \frac{\|\bfChat\|}{1-\Delta_1 \|\bfChat\|}$.

As for claim c), the difference is given by
\[
\bfmhat_\tloc-\bfmhat=(\bfChat_\tloc \bfC_\tloc^{-1}-\bfChat \bfC^{-1})\bfm+(\bfChat_\tloc\bfH_\tloc^T-\bfChat\bfH^T)\bfR ^{-1}\bfy.
\]
The norm of the matrix $(\bfChat_\tloc\bfH_\tloc^T-\bfChat\bfH^T)\bfR ^{-1}$ 
can be bounded using claim b)
\begin{align*}
\|(\bfChat_\tloc\bfH_\tloc^T-\bfChat\bfH^T)\bfR ^{-1}\|&\leq  \|\bfChat-\bfChat_\tloc\|\|\bfH^T \bfR ^{-1}\|
+\|\bfChat_\tloc\|\|\bfH_\tloc-\bfH\| \|\bfR^{-1}\|\\
&\leq  \frac{\|\bfChat\|\|\bfR^{-1}\|}{1-\Delta_1 \|\bfChat\|}(\|\bfChat\|\Delta_1+\delta_H ).
\end{align*}
As for the norm of the matrix $\bfChat_\tloc \bfC_\tloc^{-1}-\bfChat \bfC^{-1}$, we use
\[
\|\bfChat_\tloc \bfC_\tloc^{-1}-\bfChat \bfC^{-1}\|
\leq \|(\bfC_\tloc)^{-1}-\bfC^{-1}\|\|\bfChat_\tloc\|+\|\bfChat_\tloc-\bfChat\|\|\bfC^{-1}\|.
\]
From part  (b), we have
\[
\|\bfC_\tloc^{-1}-\bfC^{-1}\|\|\bfChat_\tloc\|\leq \frac{\|\bfC^{-1}\|^2 \|\bfChat\|\delta_C }{(1-\delta_C \|\bfC^{-1}\|)(1-\Delta_1 \|\bfChat\|)},
\quad\|\bfChat_\tloc-\bfChat\|\|\bfC^{-1}\|\leq \frac{\|\bfC^{-1}\|\|\bfChat\|^2\Delta_1}{1-\Delta_1 \|\bfChat\|}.
\]
In summary:
\begin{align*}
\|\bfmhat_\tloc-\bfmhat\|&\leq \left(\frac{\|\bfC^{-1}\|^2 \|\bfChat\|\delta_C }{(1-\delta_C \|\bfC^{-1}\|)(1-\Delta_1 \|\bfChat\|)}+\frac{\|\bfC^{-1}\|\|\bfChat\|^2\Delta_1}{1-\Delta_1 \|\bfChat\|}\right) \|\bfm\|\\
&\quad\quad+\frac{\|\bfChat\|\|\bfR^{-1}\|}{1-\Delta_1 \|\bfChat\|}(\|\bfChat\|\Delta_1+\delta_H ) \| \|\bfy\|. 
\end{align*}

\end{proof}

\begin{prop}
\label{prop:locprecision}
Let $\bfOmega$ be a prior precision matrix
and $\bfH$ be an observation matrix.
Let  $\delta_\Omega$ and $\delta_H$ be as defined in \eqref{eq:deltaOmega} and \eqref{eq:deltaH}. Let $\bfOmega_\tloc$ and $\bfH_\tloc$
be the localized precision and observation matrices, then
\begin{enumerate}[a)]
\item 
If $\delta_\Omega\leq \|\bfOmega\|,\Delta_2\leq \|\bfChat\|^{-1}$,
the localization creates a small perturbation of the posterior precision and covariance matrix
in the sense that $\|\bfChat^{-1}-\bfChat^{-1}_\tlocp\|\leq \Delta_2$ and
\[
\|\bfChat-\bfChat_\tlocp\|\leq \frac{\|\bfChat\|^2\Delta_2}{1-\Delta_2 \|\bfChat\|},\quad\text{where}\quad \Delta_2:=\delta_\Omega+(2\delta_H+\delta_H^2) \|\bfR^{-1}\|.
\]
\item 
Under the same conditions as in a),
the localization creates a small perturbation of the 
posterior mean in the sense that
\[
\|\bfmhat_\tlocp-\bfmhat\|\leq  \frac{ \|\bfChat\|\delta_\Omega+\|\bfOmega\|\|\bfChat\|^2\Delta_2 }{1-\Delta_2 \|\bfChat\|} \|\bfm\|+\frac{\|\bfChat\|\|\bfR^{-1}\|}{1-\Delta_2 \|\bfChat\|}(\|\bfChat\|\Delta_2+\delta_H )\|\bfy\|. 
\]
\end{enumerate}
\end{prop}

\begin{proof}
The proofs are similar to the proofs of Proposition \ref{prop:loc}. 
To prove a), we follow the proof of Proposition \ref{prop:loc} a)
and find that $\|\bfOmega-\bfOmega_\tloc\|\leq \delta_\Omega$,
and 
\eqref{tmp:HRH}
\[
\|\bfH^T \bfR^{-1} \bfH-\bfH^T_\tloc \bfR^{-1} \bfH_\tloc\|\leq (2\delta_H+\delta_H^2) \|\bfR^{-1}\|. 
\]
Therefore
\[
\|(\bfOmega+\bfH^T \bfR^{-1} \bfH)-(\bfOmega_\tloc+\bfH^T_\tloc \bfR^{-1} \bfH_\tloc)\|\leq \delta_\Omega+(2\delta_H+\delta_H^2) \|\bfR^{-1}\|=:\Delta_2.
\]
We apply Lemma \ref{lem:perturbinverse} to $\bfChat=(\bfOmega+\bfH^T \bfR^{-1} \bfH)^{-1}$, 
\[
\|\bfChat-\bfChat_\tlocp\|\leq \frac{\|\bfChat\|^2\Delta_2}{1-\Delta_2 \|\bfChat\|}.
\]
As a consequence, we have that $\|\bfChat_\tlocp\|\leq \frac{\|\bfChat\|}{1-\Delta_2 \|\bfChat\|}$.

To prove claim b), we write
\[
\bfmhat_\tlocp-\bfmhat=(\bfChat_\tlocp \bfOmega_\tloc-\bfChat \bfOmega)\bfm+(\bfChat_\tlocp\bfH_\tloc^T-\bfChat\bfH^T)\bfR ^{-1}\bfy.
\]
The norm of the matrix $(\bfChat_\tlocp\bfH_\tloc^T-\bfChat\bfH^T)\bfR ^{-1}$ 
can be bounded directly using claim a)
\begin{align*}
\|(\bfChat_\tlocp\bfH_\tloc^T-\bfChat\bfH^T)\bfR ^{-1}\|&\leq  \|\bfChat-\bfChat_{\tlocp}\|\|\bfH^T \bfR ^{-1}\|
+\|\bfChat_\tlocp\|\|\bfH_\tloc-\bfH\| \|\bfR^{-1}\|\\
&\leq  \frac{\|\bfChat\|\|\bfR^{-1}\|}{1-\Delta_2 \|\bfChat\|}(\|\bfChat\|\Delta_2+\delta_H ).
\end{align*}
The norm of the matrix $\bfChat_\tlocp \bfOmega_\tloc-\bfChat \bfOmega$ can be bounded by
\[
\|\bfChat_\tlocp \bfOmega_\tloc-\bfChat \bfOmega\|
\leq \|\bfOmega_\tloc-\bfOmega\|\|\bfChat_\tlocp\|+\|\bfChat_\tlocp-\bfChat\|\|\bfOmega\|.
\]
From part (a), we have
\[
\|\bfOmega_\tloc-\bfOmega\|\|\bfChat_\tlocp\|\leq \frac{ \|\bfChat\|\delta_\Omega }{1-\Delta_2 \|\bfChat\|},
\quad\|\bfChat_\tlocp-\bfChat\|\|\bfOmega\|\leq \frac{\|\bfOmega\|\|\bfChat\|^2\Delta_2}{1-\Delta_2 \|\bfChat\|}.
\]
In summary:
\[
\|\bfmhat_\tlocp-\bfmhat\|\leq  \frac{ \|\bfChat\|\delta_\Omega+\|\bfOmega\|\|\bfChat\|^2\Delta_2 }{1-\Delta_2 \|\bfChat\|} \|\bfm\|+\frac{\|\bfChat\|\|\bfR^{-1}\|}{1-\Delta_2 \|\bfChat\|}(\|\bfChat\|\Delta_2+\delta_H )\|\bfy\|. 
\]
\end{proof}

\section{Convergence rates for Gibbs samplers}
\label{sec:proofsGibbs}
In this section, we prove Theorems \ref{thm:gsnsimple} and \ref{thm:gsnsimplePrecision}.

\subsection{Review of Gibbs and Gauss-Seidel}
\label{sec:GibbsGauss}
Our strategy for proving these theorems relies on the 
Gauss-Seidel operator associated  with the Gibbs sampler. 
The connection between the Gibbs sampler and the 
Gauss-Seidel operator is  well documented \cite{GS89, RS97,GG01},
but we briefly review it here so that our paper can be read and understood independently.  

We consider sampling $\bfx\sim\mathcal{N}(\bfm,\bfOmega^{-1})$ using the Gibbs sampler of block size $q$, assuming $\bfOmega$ is of dimension $mq\times mq$. We will use $\bfOmega_{i,j}$ to denote its $(i,j)$-$q\times q$ sub-block, which should not be confused with the matrix entry $[\bfOmega]_{i,j}$.  We say an $mq\times mq$ matrix $\bfOmega$ is block-strictly-lower-triangular (BSLT), if $\bfOmega_{i,j}$ is nonzero only if $i>j$. We define block-diagonal and block-strictly-upper-triangular (BSUT) in an analogous way. 
Let $\bfOmega=\lowO+\diagO+\upO$ be the 
BSLT + block-diagonal + BSUT
decomposition of the precision matrix $\bfOmega$. 
In other words, $\lowO$ consists of the blocks $\{\bfOmega_{i,j}\}_{i>j}$ of $\bfOmega$,  $\diagO$ consists of the diagonal blocks of $\bfOmega$, and $\upO=\lowO^T$. 

We investigate how the $j$-th coordinate is updated at the $k+1$-th iteration. 
For simplicity, let us write the state before this coordinate update as
\[
(\bfz_1,\cdots,\bfz_m)=(\bfx^{k+1}_1,\cdots,\bfx^{k+1}_{j-1},\bfx^{k}_{j},\cdots,\bfx^{k}_{m}).
\]
Then $\log(p(\bfz))$, ignoring a constant term, can be written as
\begin{align*}
-\frac{1}{2}(\bfz-\bfm)^T \bfOmega (\bfz-\bfm)&=-\frac{1}{2}\sum_{i,i'} (\bfz_i-\bfm_i)^T\bfOmega_{i,i'} (\bfz_{i'}-\bfm_{i'})\\
&= -\frac{1}{2} (\bfz_j-\bfm_j)^T \bfOmega_{j,j}(\bfz_j-\bfm_j)-\frac{1}{2} (\bfz_j-\bfm_j)^T\sum_{i\neq j}\bfOmega_{j,i}(\bfz_{i}-\bfm_{i})\\
&\quad -\frac{1}{2} \sum_{i\neq j}(\bfz_{i}-\bfm_{i})^T \bfOmega_{i,j} (\bfz_j-\bfm_j)+F_j'(\bfz)\\
&=-\frac12 (\bfz_j-\bfm_j')^T\bfOmega_{j,j}(\bfz_j-\bfm_j')^T+F_j(\bfz).
\end{align*}
Here $F_j$ and $F_j'$ are functions independent of $\bfz_j$, and 
\[
\bfm'_j:=\bfm_j-\sum_{i\neq j} \bfOmega_{j,j}^{-1}\bfOmega_{j,i}(\bfz_i-\bfm_i)=\bfm_j-\sum_{i<j} \bfOmega_{j,j}^{-1}\bfOmega_{j,i}(\bfx^{k+1}_i-\bfm_i)-\sum_{i>j} \bfOmega_{j,j}^{-1}\bfOmega_{j,i}(\bfx^{k}_i-\bfm_i). 
\]
Using Bayes' formula,
we find that the probability of $\bfx^{k+1}_j$
conditioned on $\bfx^{k+1}_i$, $i<j$ and $\bfx^{k}_i$, $i>j$
is proportional to 
$ \exp(-\frac12 (\bfz_j-\bfm_j')^T\bfOmega_{j,j}(\bfz_j-\bfm_j'))$.
In other words, the updated coordinate has the distribution
\begin{equation}
\label{tmp:xstep}
\bfx^{k+1}_{j}\sim \mathcal{N}\left(\bfm_j-\sum_{i<j} \bfOmega_{j,j}^{-1}\bfOmega_{j,i}(\bfx^{k+1}_i-\bfm_i)-\sum_{i>j} \bfOmega_{j,j}^{-1}\bfOmega_{j,i}(\bfx^{k}_i-\bfm_i),\bfOmega_{j,j}^{-1} \right).
\end{equation}
One way to obtain a sample of this distribution is to find the solution, $\bfx^{k+1}_j$, 
of the linear equation
\begin{equation}
\label{tmp:line}
\bfOmega_{j,j} (\bfx^{k+1}_{j}-\bfm_j)+\sum_{i<j} \bfOmega_{j,i} (\bfx^{k+1}_{i}-\bfm_i)+\sum_{i>j} \bfOmega_{j,i} (\bfx^{k}_{i}-\bfm_i)=\xi^{k+1}_{j}
\end{equation}
where $\xi^{k+1}_{j}$ is an independent sample from $\mathcal{N}(0, \bfOmega_{j,j})$.

Let $\bfOmega=\lowO+\diagO+\upO$ be the BSLT + 
block-diagonal + BSUT decomposition of the precision matrix $\bfOmega$.
If we concatenate  \eqref{tmp:line} for all coordinates $j$, 
the equation becomes
\[
(\lowO+\diagO)(\bfx^{k+1}-\bfm)+ \upO (\bfx^k-\bfm)=\xi^{k+1},
\]
where $\xi^{k+1}$ is the concatenation of $\xi^{k+1}_j$, 
and distributed as $\xi^{k+1}\sim \mathcal{N}(0, \diagO)$. 
An equivalent representation is 
\begin{equation}
\label{tmp:covariance1}
(\bfx^{k+1}-\bfm)=-(\lowO+\diagO)^{-1}\upO  (\bfx^k-\bfm)+(\lowO+\diagO)^{-1}\xi^{k+1}.
\end{equation}
The correlation between two consecutive iterations then is determined by 
\[
	\Gauss:=-(\lowO+\diagO)^{-1}\upO.
\]	 
The spectral radius of $\Gauss$ decides how fast  the Gibbs sampler converges in $l_2$ norm. 

The matrix $\Gauss$ is known as the ``Gauss-Seidel operator,''
which is also used in the iterative solution of linear equations. 
Recall that 
$\mathcal{N}(0, \bfC)$ is the invariant distribution of 
$\bfx^k-\bfm$.
By comparing covariances on both sides of \eqref{tmp:covariance1}, we find that
\begin{equation}
\label{eqn:invariant}
\bfC=\Gauss \bfC \Gauss^T+(\lowO+\diagO)^{-1} \diagO (\lowO+\diagO)^{-T}.
\end{equation}
It follows that $\bfC\preceq \Gauss \bfC \Gauss^T$, which in turn implies that the spectral radius of $\Gauss$ is less than $1$, which implies convergence.
Here and below, for two symmetric matrices $\bfA$ and $\bfB$, we use $\bfA\preceq \bfB$ to indicate that $\bfB-\bfA$ is positive semidefinite. However, in order to show  the convergence rate is dimension-independent, we need to exploit the banded structure of $\bfC$ or~$\bfOmega$, which will be the purpose of the next section.

\subsection{Gauss-Seidel with localized structures}
First we need two estimates.

\begin{lem}
\label{lem:matrixbasic}
For any positive definite $qm\times qm $ matrix $\bfC$, denote its maximum eigenvalue as $\lambda_\text{max}$, its minimum eigenvalue as $\lambda_\text{min}$, its condition number as $\Cond=\lambda_\text{max}/\lambda_\text{min}$, its inverse $\bfOmega=\bfC^{-1}$. Then the $q\times q$ blocks satisfy:
\begin{equation}
\label{eq:Inequalities}
\lambda_\text{min} \bfI\preceq \bfC_{i,i}\preceq \lambda_\text{max} \bfI,\quad
\lambda^{-1}_\text{max}\bfI\preceq \bfOmega_{i,i}\preceq \lambda^{-1}_\text{min} \bfI,
\end{equation}
\begin{equation}
\label{eq:Inequalities2}
\|\bfOmega^{-1/2}_{i,i}\bfOmega_{i,j}\bfOmega^{-1/2}_{j,j}\|\leq 1-\Cond^{-1},\quad
\|\bfC^{-1/2}_{i,i}\bfC_{i,j}\bfC^{-1/2}_{j,j}\|\leq 1-\Cond^{-1}.
\end{equation}
 
\end{lem}
\begin{proof}
Let $\lambda_i$ be an eigenvalue of $\bfC_{i,i}$ and $v\in \reals^q$ be one of its eigenvectors with norm $1$. Let $\textbf{v}$ be the $\reals^{qm}$ vector with its $i$-th block being $v$. Then 
\[
\lambda_i=v^T\bfC_{i,i}v= \textbf{v}^T\bfC\textbf{v} \in [\lambda_\text{min},\lambda_\text{max}]. 
\] 
The left inequality in~(\ref{eq:Inequalities}) follows. The right inequality of equation~(\ref{eq:Inequalities}) can be derived in a similar fashion.

Let $\bfx$ and $\bfy$ be the left and right singular vectors corresponding 
to the largest singular value of 
$\Gamma_{i,j}=\bfOmega^{-1/2}_{i,i}\bfOmega_{i,j}\bfOmega^{-1/2}_{j,j}$. 
The vectors $\bfx$ and $\bfy$ are of dimension $q$, 
have norm one,  $\|\bfx\|_2=\|\bfy\|_2=1$, 
and $\bfx^T \Gamma_{i,j}\bfy =\|\Gamma_{i,j}\|$. 
Now consider  an $qm$ dimensional vector $\bfv$, 
where its $i$-th block is $\bfOmega_{i,i}^{-1/2}\bfx$,  
its $j$-th block is $-\bfOmega_{j,j}^{-1/2}\bfy$, 
and all other blocks are zero.  
Then 
\begin{align*}
\bfv^T \bfOmega\bfv&=\bfx^T\bfOmega_{i,i}^{-1/2}\bfOmega_{i,i}\bfOmega_{i,i}^{-1/2}\bfx-2\bfx^T\bfOmega_{i,i}^{-1/2}\bfOmega_{i,j}\bfOmega_{j,j}^{-1/2}\bfy+
\bfy^T\bfOmega_{j,j}^{-1/2}\bfOmega_{j,j}\bfOmega_{j,j}^{-1/2}\bfy\\
&=2-2 \bfx^T \Gamma_{i,j}\bfy=2(1-\|\Gamma_{i,j}\|). 
\end{align*}
On the other hand,  
\[
\|\bfv\|^2=\|\bfOmega_{i,i}^{-1/2}\bfx\|^2+\|\bfOmega_{j,j}^{-1/2}\bfy\|^2\geq 2\lambda_\text{min}.
\]
Thus 
\[
2(1-\|\Gamma_{i,j}\|)=2\bfv^T \bfOmega\bfv\geq 2\lambda_\text{min} \lambda_\text{max}^{-1}=2\Cond^{-1}.
\]
The left inequality in~(\ref{eq:Inequalities2}) follows. Since the derivation above uses nothing of $\bfOmega$ other than its eigenvalues, so the right inequality in~(\ref{eq:Inequalities2}) also holds.
\end{proof}

For our proofs below,
we need the following bound for an operator norm. 
For $q=1$, this bound appeared in \cite{BL08} as inequality (A2).
Note that this bound is well-suited for block-sparse $\bfA$ 
since then the right hand side consists of only a few terms.  

\begin{lem}
\label{lem:norm}
For any $qm\times qm$ matrix $\bfA$, the following holds
\[
\|\bfA\| \leq \left(\max_{i=1,\cdots,m} \sum_{j=1}^m \|\bfA_{i,j}\|\right)^{1/2} \left(\max_{i=1,\cdots,m} \sum_{j=1}^m \|\bfA_{j,i}\|\right)^{1/2}.
\]
\end{lem}
\begin{proof}
First we show the claim for symmetric $\bfA=\bfA^T$.
In this case, the bound for the norm of $\bfA$ is
\begin{equation}
\label{tmp:symmetric}
\|\bfA\| \leq \max_{i=1,\cdots,m} \sum_{j=1}^m \|\bfA_{i,j}\|. 
\end{equation}
Let $\bfv=[\bfv_1,\ldots, \bfv_m]\in \reals^{qm}$ be an eigenvector so that $\lambda \bfv=\bfA\bfv$ while $|\lambda|=\|\bfA\|$. 
Suppose $\|\bfv_{i_*}\|=\max_i \{\|\bfv_i\|\}$, then 
\[
|\lambda|\|\bfv_{i_*}\|=\|\lambda\bfv_{i_*}\|=\left\|\sum_{j=1}^m \bfA_{i_*,j}\bfv_{j}\right\|\leq \sum_{j=1}^m \|\bfA_{i_*,j}\|\|\bfv_{j}\|\leq \|\bfv_{i_*}\| \sum_{j=1}^m \|\bfA_{i_*,j}\|.
\]
This leads to $\|\bfA\|\leq  \sum_{j=1}^m \|\bfA_{i_*,j}\|$ and hence \eqref{tmp:symmetric}. 

For a general $\bfA$, note that $\|\bfA\|=\|\bfA^T \bfA\|^{1/2}$. 
The matrix $\bfP=\bfA^T\bfA$ is symmetric, 
and its blocks are 
\[
\bfP_{i,j}=\sum_k \bfA_{k,i}^T\bfA_{k,j}.
\]
Applying \eqref{tmp:symmetric} to $\bfP$, we obtain
\begin{align*}
\|\bfP\|\leq \max_i \sum_{j=1}^m\|\bfP_{i,j}\|&\leq  \max_i \sum_j \sum_k \|\bfA_{k,i}^T\bfA_{k,j}\|\\
&\leq  \max_i \sum_k \sum_j \|\bfA_{k,i}\|\|\bfA_{k,j}\|\\
&\leq  \max_i \sum_k  \|\bfA_{k,i}\| \sum_j\|\bfA_{k,j}\|\\
&\leq \left(\max_i \sum_k  \|\bfA_{k,i}\|\right)\left(\max_i \sum_j  \|\bfA_{i,j}\|\right).
\end{align*}
This leads to our general claim. 
\end{proof}

Now we are at the position to establish bounds for the Gauss Seidel operator.
\begin{lem}
\label{lem:sparseprecision}
If $\bfOmega$ is block-tridiagonal with $\Cond$ being its condition number, then 
\[
 \Gauss \bfC\Gauss^T\preceq \frac{\Cond(1-\Cond^{-1})^2}{1+\Cond(1-\Cond^{-1})^2} \bfC. 
 \]
\end{lem}

\begin{proof}
Since $\bfOmega$ is block-tridiagonal, 
$\upO$ has at most one  nonzero block in each row and each column, 
and, likewise, $\diagO^{-1/2} \upO\diagO^{-1/2}$ has 
at most one  nonzero block in each row and each column. 
Therefore, by Lemma \ref{lem:norm}
\[
\left\|\diagO^{-1/2} \upO\diagO^{-1/2}\right\|\leq \max_{i=1,\ldots,m-1} \left\|\bfOmega_{i,i}^{-1/2}\bfOmega_{i,i+1}\bfOmega^{-1/2}_{i+1,i+1}\right\|, 
\] 
which by Lemma \ref{lem:matrixbasic} is bounded by $1-\Cond^{-1}$. 

To continue, we look at the right hand side of \eqref{eqn:invariant}. 
We want to show that for some $\gamma>0$,
\begin{equation}
\label{eqn:noise}
\Gauss \bfC\Gauss^T\preceq  \gamma(\lowO+\diagO)^{-1}\diagO(\lowO+\diagO)^{-T}
\end{equation}
For this purpose, note that 
\[
\Gauss \bfC\Gauss^T=(\lowO+\diagO)^{-1}\upO\bfC \upO^T (\lowO+\diagO)^{-T},
\]
and that
\begin{align*}
\upO\bfC \upO^T=\diagO^{1/2}(\diagO^{-1/2} \upO\diagO^{-1/2}) (\diagO^{1/2} \bfC\diagO^{1/2} ) (\diagO^{-1/2} \upO\diagO^{-1/2})^T\diagO^{1/2}.
\end{align*}
Thus, in order to prove \eqref{eqn:noise}, 
it suffices to find a $\gamma$ such that 
\[
\left\|(\diagO^{-1/2} \upO\diagO^{-1/2}) (\diagO^{1/2} \bfC\diagO^{1/2} ) (\diagO^{-1/2} \upO\diagO^{-1/2})^T\right\|\leq \gamma. 
\]
By Lemma \ref{lem:matrixbasic}, this is straight forward,
since we have that
\begin{align*}
&\left\|(\diagO^{-1/2} \upO\diagO^{-1/2}) (\diagO^{1/2} \bfC\diagO^{1/2} ) (\diagO^{-1/2} \upO\diagO^{-1/2})^T\right\|\\
&\quad\leq (1-\Cond^{-1})^2\|\diagO\| \|\bfC\|\leq \Cond(1-\Cond^{-1})^2=:\gamma,
\end{align*}
since $\diagO$ is diagonal with blocks bounded in operator norm by $\lambda_\text{min}^{-1}$. 

Finally, we can plug \eqref{eqn:noise} into \eqref{eqn:invariant}, and find that 
\[
\bfC=\Gauss \bfC\Gauss^T+(\lowO+\diagO)^{-1}\diagO(\lowO+\diagO)^{-T}\succeq (\gamma^{-1}+1) \Gauss \bfC\Gauss^T. 
\]
\end{proof}

\begin{lem}
\label{lem:sparsecov}
If $\bfC$ is block-tridiagonal with $\Cond$ being its condition number, then 
\[
\Gauss \bfC\Gauss^T\preceq 
\frac{2(1-\Cond^{-1})^2\Cond^4}{1+2(1-\Cond^{-1})^2\Cond^4} \bfC. 
\]
\end{lem}
\begin{proof}
Since $\bfOmega^{-1}=\bfC$, we apply the Woodbury's formula to $(\lowO+\diagO)^{-1}=(\bfOmega-\upO)^{-1}$,
\[
(\lowO+\diagO)^{-1}=(\bfOmega-\upO)^{-1}=\bfC+ \bfC\upO(\bfI-\bfC\upO)^{-1}\bfC=(\bfI-\bfC\upO)^{-1}\bfC. 
\]
Consequentially, $\Gauss=(\bfI-\bfC\upO)^{-1}\bfC\upO$.

Next, we claim that  $\bfC\upO$ is BUT, 
and that only the blocks $(\bfC\upO)_{i,i},(\bfC\upO)_{i,i+1}$ are nonzero. 
To see it is BUT, recall that $\bfC$ is block-tridiagonal, 
$(\upO)_{i,j}=\bfOmega_{i,j}\unit_{i\leq j-1}$, 
\[
(\bfC\upO)_{i,j}=\bfC_{i,i-1} \bfOmega_{i-1,j}\unit_{i\leq j}+\bfC_{i,i} \bfOmega_{i,j}\unit_{i\leq j-1}+\bfC_{i,i+1} \bfOmega_{i+1,j}\unit_{i\leq j-2}. 
\]
Thus $(\bfC\upO)_{i,j}$ is nonzero only if $i\leq j$, 
i.e., $\bfC\upO$ is BUT. 
Moreover, if $j\geq i+2$, by the identity $\bfC\bfOmega=\bfI$, we have
\[
\mathbf{0}=(\bfC\bfOmega)_{i,j}=\bfC_{i,i-1} \bfOmega_{i-1,j}+\bfC_{i,i} \bfOmega_{i,j}+\bfC_{i,i+1} \bfOmega_{i+1,j}=(\bfC\upO)_{i,j},
\]
proving our claim.

Next, not that for $j=i$, we have
\[
(\bfC\upO)_{i,i}=\bfC_{i,i-1} \bfOmega_{i-1,i}\unit_{i\leq i}+\bfC_{i,i} \bfOmega_{i,i}\unit_{i\leq i-1}+\bfC_{i,i+1} \bfOmega_{i+1,i}\unit_{i\leq i-2}=\bfC_{i,i-1} \bfOmega_{i-1,i}.
\]
Likewise, for $j=i+1$, we have 
\[
\mathbf{0}=(\bfC\bfOmega)_{i,i+1}=\bfC_{i,i-1} \bfOmega_{i-1,i+1}+\bfC_{i,i} \bfOmega_{i,i+1}+\bfC_{i,i+1} \bfOmega_{i+1,i+1}=(\bfC\upO)_{i,i+1}+ \bfC_{i,i+1} \bfOmega_{i+1,i+1}.
\]
In other words, $(\bfC\upO)_{i,i+1}=-\bfC_{i,i+1} \bfOmega_{i+1,i+1}$.

Applying Lemma \ref{lem:norm}  leads to
\begin{align*}
\|\bfC\upO\|\leq \left(\max_{i=1,\cdots,m} \{ \|(\bfC\upO)_{i,i}\|+\|(\bfC\upO)_{i,i+1}\|\}\right)^{1/2} \left(\max_{i=1,\cdots,m}   \|(\bfC\upO)_{i,i}\|\right)^{1/2}.
\end{align*}
Applying Lemma \ref{lem:matrixbasic} to $(\bfC\upO)_{i,i}=\bfC_{i,i-1} \bfOmega_{i-1,i}$ gives 
\begin{align*}
\|(\bfC\upO)_{i,i}\|&\leq \|\bfC_{i,i}^{1/2}\|\|\bfC_{i,i}^{-1/2}\bfC_{i,i-1} \bfC_{i-1,i-1}^{-1/2}\|\|\bfC_{i-1,i-1}^{1/2}\|\|\bfOmega_{i-1,i-1}^{1/2}\|
\| \bfOmega_{i-1,i-1}^{-1/2}\bfOmega_{i-1,i}\bfOmega_{i,i}^{-1/2}\| \|\bfOmega_{i,i}^{1/2}\|\\
&\leq (1-\Cond^{-1})^2\Cond\leq  (1-\Cond^{-1})\Cond.
\end{align*}
Similarly, $(\bfC\upO)_{i,i+1}=-\bfC_{i,i+1} \bfOmega_{i+1,i+1}$,
which by Lemma \ref{lem:matrixbasic} implies that
\begin{align*}
\|(\bfC\upO)_{i,i+1}\|&\leq \|\bfC_{i,i}^{1/2}\|\|\bfC_{i,i}^{-1/2}\bfC_{i,i+1} \bfC_{i+1,i+1}^{-1/2}\|\|\bfC_{i+1,i+1}^{1/2}\|\|\bfOmega_{i+1,i+1}\|
\leq (1-\Cond^{-1})\Cond. 
\end{align*}
Consequentially, another application of Lemma \ref{lem:norm} implies that
\[
\|\bfC\upO\|\leq \sqrt{2}(1-\Cond^{-1})\Cond. 
\]

To continue,  
we again want to use \eqref{eqn:invariant} by exploiting relations like \eqref{eqn:noise}. 
We first note that
\begin{align*}
\Gauss \bfC\Gauss^T&=(\bfI-\bfC\upO)^{-1}\bfC\upO\bfC  \upO^T \bfC(\bfI-\bfC\upO)^{-T}, \\
(\lowO+\diagO)^{-1}\diagO(\lowO+\diagO)^{-T}&=(\bfI-\bfC\upO)^{-1}\bfC \diagO\bfC (\bfI-\bfC\upO)^{-T}. 
\end{align*}
Using $\|\bfC\upO\|\leq \sqrt{2}(1-\Cond^{-1})\Cond$, we have
\[
\bfC\upO\bfC  \upO^T \bfC\preceq 2(1-\Cond^{-1})^2\Cond^2 \lambda_\text{max} \bfI. 
\]
Moreover, 
\[
\bfC \diagO\bfC\succeq \lambda_\text{max}^{-1}\bfC \bfC \bfI\succeq \lambda_\text{max}^{-1}\lambda_\text{min}^{2}\bfI. 
\]
Consequentially, $\bfC\upO\bfC  \upO^T \bfC\preceq 2(1-\Cond^{-1})^2\Cond^4 \bfC \diagO\bfC$ and
\[
\Gauss \bfC\Gauss^T\preceq 2(1-\Cond^{-1})^2\Cond^4(\lowO+\diagO)^{-1}\diagO(\lowO+\diagO)^{-T}.
\]
Combining the above inequality with  \eqref{eqn:invariant} leads to
\[
\bfC=\Gauss \bfC\Gauss^T+(\lowO+\diagO)^{-1}\diagO(\lowO+\diagO)^{-T}\succeq ((2(1-\Cond^{-1})^2\Cond^4)^{-1}+1) \Gauss \bfC\Gauss^T. 
\]
\end{proof}

\subsection{Proofs of the main theorems} Armed with these results, the proofs for the main theorems follow from an elementary coupling argument.
\begin{proof}[Proof of Theorems \ref{thm:gsnsimple} and \ref{thm:gsnsimplePrecision}]
As discussed above, and illustrated in Section \ref{sec:GibbsGauss}, 
we can generate iterates from the Gibbs sampler with block-size $q$ by
solving the linear equations
\[
(\bfx^{k+1}-\bfm)=\Gauss (\bfx^k-\bfm)+(\lowO+\diagO)^{-1}\xi^{k+1}\quad k=0,1,\cdots, 
\]
where $\xi^{k+1}$ are i.i.d. samples from $\mathcal{N}(0, \diagO)$. 

Next we consider a random sample $\bfz^0$ from $\mathcal{N}(\bfm, \bfC)$. We can apply the block-Gibbs sampler with $\bfz^0$ as the initial condition, while using the same sequence $\xi^k$. In other words, we generate $\bfz^{k+1}$ by letting
\[
(\bfz^{k+1}-\bfm)=\Gauss (\bfz^k-\bfm)+(\lowO+\diagO)^{-1}\xi^{k+1}\quad k=0,1,\cdots.
\]
Since $\mathcal{N}(\bfm, \bfC)$ is the invariant measure for the Gibbs sampler, marginally $\bfz^k\sim \mathcal{N}(\bfm,\bfC)$. 

Next we look at the difference between the two Gibbs samplers, $\Delta^k=\bfx^k-\bfz^k$. 
Note that 
\[
\Delta^0\sim\mathcal{N}(\bfx^0-\bfm, \bfC),\quad \Delta^{k+1}=\Gauss \Delta^k,\quad k=0,1\cdots. 
\]
Consequentially, $\Delta^k\sim \mathcal{N}(\Gauss^k\Delta_0, \Gauss^k\bfC (\Gauss^T)^k)$, where $\Delta_0:=\bfx^0-\bfm$.  Since
\[
\Delta^0\Delta^{0T}=\bfC^{1/2} (\bfC^{-1/2} \Delta_0) (\bfC^{-1/2} \Delta_0)^T\bfC^{1/2}\preceq  \|\bfC^{-1/2} \Delta_0\|^2\bfC,
\]
we find that 
\begin{align}
\E\|\bfC^{-1/2}\Delta^k\|^2&=\|\bfC^{-1/2}\Gauss^k\Delta_0\|^2+ \text{tr}(\bfC^{-1/2} \Gauss^k \bfC  (\Gauss^T)^k \bfC^{-1/2} )\nonumber\\
&=\text{tr}(\bfC^{-1/2}\Gauss^k\Delta_0\Delta^{T}_0 (\Gauss^T)^k \bfC^{-1/2}+\bfC^{-1/2} \Gauss^k \bfC  (\Gauss^T)^k \bfC^{-1/2} )\nonumber\\
\label{eq:FinalInequality}
&\leq (1+\|\bfC^{-1/2} \Delta_0\|^2)\cdot\text{tr}(\bfC^{-1/2} \Gauss^k \bfC  (\Gauss^T)^k \bfC^{-1/2} ).
\end{align}

If $\bfOmega$ is block-tridiagonal,
then, by Lemma \ref{lem:sparseprecision},
\[
	\bfC^{-1/2} \Gauss^k \bfC  (\Gauss^T)^k \bfC^{-1/2}\leq \beta^k \bfI,\quad
	\beta = \frac{\Cond(1-\Cond^{-1})^2}{1+\Cond(1-\Cond^{-1})^2}.
\]
Combining the above inequality with~(\ref{eq:FinalInequality})
proves Theorem \ref{thm:gsnsimplePrecision}.

If $\bfC$ is block-tridiagonal,
then, by Lemma \ref{lem:sparsecov},
\[
	\bfC^{-1/2} \Gauss^k \bfC  (\Gauss^T)^k \bfC^{-1/2}\leq \beta^k \bfI,\quad
	\beta = \frac{2(1-\Cond^{-1})^2\Cond^4}{1+2(1-\Cond^{-1})^2\Cond^4}.
\]
Combining the above inequality with~(\ref{eq:FinalInequality})
proves Theorem \ref{thm:gsnsimple}.
\end{proof}

\bibliographystyle{siamplain}
\bibliography{references}
\end{document}